\documentclass{fundam}

\publyear{}
\papernumber{}
\volume{}
\issue{}

\usepackage[T1]{fontenc}
\usepackage{graphicx}

\usepackage{bcprules}
\usepackage{proof}
\usepackage{myproof}
\usepackage{paralist}
\usepackage{multicol}
\usepackage{color}
\usepackage{comment}

\usepackage{stmaryrd}
\usepackage{amsmath}
\usepackage{amssymb}
\newenvironment{pfof}[1]
  {\trivlist\PRstyle\item[]{\bfseries Proof of #1:}\newline}{\QED\endtrivlist}
\newcommand\qed{\QED}

\newif\ifdraft\draftfalse
\newif\ifshort\shortfalse
\newif\ifnarrow\narrowtrue
\allowdisplaybreaks
\ifdraft
\newcommand\nk[1]{\textcolor{red}{[#1 -nk]}}
\newcommand\asd[1]{{\color{blue}[#1 -asd]}}
\newcommand\asdm[1]{{\footnotesize {\color{blue}[memo (ignorable): #1 -asd]}}} 
\else
\newcommand\nk[1]{}
\newcommand\asd[1]{}
\newcommand\asdm[1]{}
\fi
\newcommand\penv{\seq{x}_{1,\ldots,k}}
\newcommand\toform[1]{#1^{\mu}}
\newcommand\DEFEQ{\mathbin{:=}}
\newcommand\Sum{\mathit{sum}}
\newcommand\Sumex{\mathit{ex1}}
\newcommand\evC{E}
\newcommand\hole{[\,]}
\newcommand\red{\longrightarrow}
\newcommand\reds{\red^*}
\newcommand\todisjeq[1]{\sim_{#1}}
\newcommand\todisj[1]{#1^{\#}}
\newcommand\todisjaux[1]{#1^{\#'}}
\newcommand\fromdisj[1]{#1^{\flat}}
\newcommand\fromdisjT[2]{#1^{\flat,#2}}
\newcommand\fromdisjTE[1]{#1^{\flat}}
\newcommand\fromdisjTEp[1]{#1^{\flat'}}
\newcommand\Mark{*}
\newcommand\maxar{M}
\newcommand\NT{\Theta} 

\newcommand\muhflz{$\mu$HFL(Z)}
\newcommand\hflz{HFL(Z)}
\newcommand\nuhflz{$\nu$HFL(Z)}
\newcommand\rethfl{\textsc{ReTHFL}}
\newcommand\pcsat{PCSat}
\newcommand\hoice{\textsc{HoIce}}
\newcommand\horus{Horus}
\newcommand\dual[1]{\overline{#1}}
\newcommand\decomp{\mathtt{decomp}}
\newcommand\decomparg{\mathtt{decomparg}}
\newcommand\form{\varphi}
\newcommand\formalt{\psi}
\newcommand\formaltp{\xi}
\newcommand\fty{\tau}
\newcommand\sty{\kappa}
\newcommand\INT{\mathtt{Int}}
\newcommand\Prop{\star}
\newcommand\ord{\mathtt{ord}}
\newcommand\ar{\mathtt{ar}}
\newcommand\stenv{\mathcal{K}}
\newcommand\lstenv{{\Sigma}}
\newcommand\toST[1]{{#1}\downarrow_{\INT}}
\newcommand\lsty{\sigma}
\newcommand\lrsty{\eta}
\newcommand\COL{\mathbin{:}}
\newcommand\p{\vdash}
\newcommand\pn{\p_{\NT}}
\newcommand\tr{\leadsto}

\newcommand\seq[1]{\widetilde{#1}}
\newcommand\bra[2]{#1\land #2}
\newcommand\nc{\lor}

\newcommand\gar{\mathtt{gar}}
\newcommand\pST{\p_{\mathtt{ST}}}
\newcommand\pLST{\p_{\mathtt{LST}}}
\newcommand\sem[1]{\llbracket #1 \rrbracket}
\newcommand\semd[1]{\llparenthesis #1 \rrparenthesis}
\newcommand\TRUE{\mathtt{true}}
\newcommand\FALSE{\mathtt{false}}
\newcommand\set[1]{\{#1\}}
\newcommand\LEQ{\sqsubseteq}
\newcommand\Z{\mathbf{Z}}
\newcommand\imp{\Rightarrow}
\newcommand\IFF{\Leftrightarrow}
\newcommand\GLB{\sqcap}
\newcommand\LUB{\sqcup}

\newcommand\LFP{\mathbf{LFP}}
\newcommand\dom{\mathit{dom}}

\newcommand\err{\mathbf{succ}}
\newcommand\dnt{\blacksquare}
\newcommand\ant{\square}
\newcommand\Eunit{(\,)}
\newcommand\Efix[3]{\mathbf{fix}^{#3}(#1,#2)}
\newcommand\Eass{\mathbf{assume}}
\newcommand\Eif[3]{\mathbf{if}\ #1\ \mathbf{then}\ #2\ \mathbf{else}\ #3}
\newcommand\Tunit{\mathtt{unit}}
\newcommand\tohflz[1]{#1^\dagger}
\newcommand\PlayerP{\mathtt{P}}
\newcommand\PlayerO{\mathtt{O}}
\newcommand\mochi{\textsc{MoCHi}}

\makeatletter
\RequirePackage[bookmarks,unicode,colorlinks=true]{hyperref}%
\usepackage{cleveref}
   \def\@citecolor{blue}%
   \def\@urlcolor{blue}%
   \def\@linkcolor{blue}%

\def\orcidID#1{\smash{\href{http://orcid.org/#1}{\protect\raisebox{-1.25pt}{\protect\includegraphics{ORCID_Color}}}}}
\makeatother

\crefname{theorem}{Theorem}{Theorems}
\crefname{lemma}{Lemma}{Lemmas}
\crefname{figure}{Figure}{Figures}
\crefname{corollary}{Corollary}{Corollaries}
\crefname{remark}{Remark}{Remarks} 
\crefname{remk}{Remark}{Remarks}
\crefname{definition}{Definition}{Definitions} 
\crefname{defi}{Definition}{Definitions}
\crefname{example}{Example}{Examples}
\crefname{proposition}{Proposition}{Propositions}
\crefname{section}{Section}{Sections}

\renewcommand\cref{\Cref}

\makeatletter
\newlength{\negph@wd}
\DeclareRobustCommand{\negphantom}[1]{%
  \ifmmode
    \mathpalette\negph@math{#1}%
  \else
    \negph@do{#1}%
  \fi
}
\newcommand{\negph@math}[2]{\negph@do{$\m@th#1#2$}}
\newcommand{\negph@do}[1]{%
  \settowidth{\negph@wd}{#1}%
  \hspace*{-\negph@wd}%
}
\makeatother

\newcommand{\defe}{:=}

\newcommand\rnp[1]{\ensuremath{\mathrm{\rn{#1}}}} 

\newcommand\dep[1]{\succ_{#1}}

\newcommand\tree{v}
\renewcommand\pi{\tree}

\newcommand\esubst[2]{\{#2/#1\}}
\newcommand\ess[4]{\{ #2 / #1 ,\dots, #4 / #3 \}}
\newcommand\subst[2]{[#2/#1]}
\newcommand\oss[4]{[ #2 / #1 ,\dots, #4 / #3 ]} 

\newcommand{\eqdpara}[1]{=_{#1}}
\newcommand{\eqd}{\eqdpara{D}}
\newcommand{\eqdp}{\eqdpara{D'}}

\newcommand{\reddpara}[1]{\red_{#1}}
\newcommand{\redd}{\reddpara{D}}
\newcommand{\redds}{\redd^*}

\newcommand{\nf}{\zeta}

\newcommand{\fm}{\form}
\newcommand{\fa}{\formalt}
\newcommand{\fp}{\formaltp}

\newcommand{\myitem}{\(\bullet\)\ }

\newcommand\decompA{\mathtt{decompArg}}

\newcommand\fea{\alpha}
\newcommand\feb{\beta}
\newcommand\fec{\gamma}

\newcommand\Nat{\mathbb{N}}

\newcommand\btype[2]{{#1}^{@#2}} 

\newcommand\caa[3]{\seq{#1}_{\setminus #2}^{#3}} 

\newcommand\seqe[1]{\seq{e}_{(#1)}}
\newcommand\seqz[1]{\seq{z}_{(#1)}}

\newcommand\semes[2]{\sem{(#1,#2)}}

\newcommand\envx{\penv}
\newcommand\varf[2]{#1}

\newcommand\beeq{\eqdp}

\newcommand\tags[2][]{\tag*{\((#2)\)#1}}

\newcommand\FVf{\mathrm{FV}'}

\newcommand\maar[1]{\maxar} 

\begin{document}

\title{On Higher-Order Reachability Games vs May Reachability}

\author{Kazuyuki Asada\\
  Tohoku University\\
  \and Hiroyuki Katsura\\
  The University of Tokyo\\
  \and Naoki Kobayashi\\
  The University of Tokyo}
\maketitle

\runninghead{K. Asada et al.}{Higher-Order Reachability Games vs May Reachability}

\begin{abstract}
We consider the reachability problem for higher-order functional
programs and study the relationship between reachability games (i.e.,
the reachability problem for programs with angelic and demonic
nondeterminism) and may-reachability (i.e., the reachability problem
for programs with only angelic nondeterminism).  We show that
reachability games for order-\(n\) programs can be reduced to
may-reachability problems for order-(\(n+1\)) programs, and vice versa.
We formalize the reductions by using higher-order fixpoint logic
and prove their correctness. We also discuss applications \ifshort\else of the reductions \fi
to higher-order program verification.
\end{abstract}


\section{Introduction}
\label{sec:intro}

This paper considers the reachability problem for simply-typed, call-by-name
higher-order functional programs with integers, recursion, and two kinds of non-deterministic branches
(angelic and demonic ones). The problem of solving reachability games (hereafter,
simply called the reachability game problem) asks, given a higher-order functional
program and a specific control point \(\err\) of the program,
whether there exists a sequence of choices on angelic non-determinism
that makes the program reach \(\err\)
no matter what choices are made on demonic non-determinism.
Thus, our reachability game problem is just a special case of the notion of two-player reachability
games~\cite{Automata}, where the game arena is specified as a higher-order functional program.
\ifshort\else
(An important restriction compared to the general notion of reachability games is that
each vertex may have only a finite number of outgoing edges, although there can be infinitely
many vertices.) \fi
Various program verification problems can be reduced to the reachability game problem.
For example, the termination problem, which asks whether a given program terminates for
any sequence of non-deterministic choices, is a special case of
the reachability game problem, where all the non-deterministic branches are demonic,
and all the termination points are expressed by \(\err\). The safety verification problem,
which asks whether a given program may fall into an error state after some sequence of
non-deterministic choices, is also a special case, where all the non-deterministic branches are
angelic, and error states are expressed by \(\err\).

We establish relations between the reachability game problem and the
\emph{may-reachability} problem, a special case of the reachability game problem
where all the non-deterministic choices are angelic (hence, may-reachability 
is a one-player game).
We show mutual translations between
the reachability game problem for order-\(n\) programs and
the may-reachability problem for order-(\(n+1\)) programs. (Here, the order of a program
is defined as the type-theoretic order; the order of a function that takes only integers is 0,
and the order of a function that takes an order-0 function is 1, etc.)
The translations are size-preserving in the sense that for any order-\(n\) program \(M\),
one can effectively construct an order-(\(n+1\)) program \(M'\) such that
the answer to the reachability game problem for \(M\) is the same as 
the answer to the may-reachability problem for \(M'\), and the size of \(M'\) is polynomial
in that of \(M\); and vice versa. Here, the largest size of the simple types
in the input is treated as a fixed parameter. For the order-decreasing
direction, inputs are assumed to contain no \(\beta\)-redexes;
applications of fixpoint expressions are allowed.
The precise input and size conventions are given in Section~\ref{sec:theorem}.

The translation from reachability games to may-reachability allows us
to use higher-order
program verification tools specialized to may-reachability (or, unreachability to error states)
such as \mochi{}~\cite{KSU11PLDI} and Liquid types~\cite{Jhala08}
to check a wider class of properties represented as reachability games. Conversely,
the translation from may-reachability to reachability games allows
us, for example, to use verification tools that can solve reachability games for order-0 programs,
such as CHC solvers~\cite{DBLP:journals/fmsd/KomuravelliGC16,Eldarica,DBLP:journals/jar/ChampionCKS20}
to check may-reachability of order-1 programs.

We formalize our translations for \muhflz{}, which is a fragment of
\hflz{}~\cite{KTW18ESOP} without greatest fixpoint operators and modal operators, where
\hflz{} is an extension of Viswanathan and Viswanathan's higher-order fixpoint
logic~\cite{Viswanathan04} with integers. 
The use of higher-order fixpoint formulas rather than higher-order programs in the formalization
of the translations is justified by
the result of Kobayashi et al.~\cite{KTW18ESOP,DBLP:conf/pepm/WatanabeTO019}, 
 that there is a direct correspondence between the reachability problem for higher-order programs
and the validity problem for the corresponding higher-order fixpoint formulas,
where angelic and demonic branches in programs correspond to disjunctions and conjunctions
respectively.

The rest of this paper is structured as follows.
\ifshort
\else
Section~\ref{sec:pre} introduces \muhflz{}, and clarifies the relationship between
the validity checking problem for \muhflz{} and the reachability problem
for higher-order programs.
In particular, may-reachability for order-\(n\) programs corresponds to
validity checking for order-\(n\) disjunctive \muhflz{} formulas,
while reachability games for order-\(n\) programs correspond to
validity checking for general order-\(n\) \muhflz{} formulas.
Section~\ref{sec:todisj} defines a translation from
order-\(n\) \muhflz{} formulas to
order-(\(n+1\)) disjunctive \muhflz{} formulas,
and proves its correctness.
Section~\ref{sec:fromdisj} defines a translation in the opposite direction,
and proves its correctness.
By the correspondence explained in Section~\ref{sec:pre},
these translations yield reductions in both directions between
the order-\(n\) reachability game problem and
the order-(\(n+1\)) may-reachability problem.
Section~\ref{sec:app} discusses applications of the reductions
and reports some experimental results.
Section~\ref{sec:related} discusses related work
and Section~\ref{sec:conc} concludes the paper.
This is an extended and revised version of the paper that appeared in Proceedings of RP 2022~\cite{DBLP:conf/rp/AsadaKK22}.
We have added definitions, examples and full proofs.
\fi

\section{\muhflz{} and Reachability Problems}
\label{sec:pre}
In this section, we first introduce \muhflz{}, a fragment of higher-order fixpoint logic
\hflz{}~\cite{KTW18ESOP} (which is in turn an extension of Viswanathan and Viswanathan's
higher-order fixpoint logic~\cite{Viswanathan04} with integers)
without greatest fixpoint operators. We then review the relationship between \muhflz{} and
reachability problems, and state the main theorem of this paper.
\ifshort
\else
\subsection{Syntax}
\label{sec:syntax}
\fi
The set of \emph{(simple) types}, ranged over by \(\sty\), is given by:
\ifshort
\else
\[
\begin{array}{l}
\sty \mbox{ (types)} ::= \INT \mid \fty\\
\fty \mbox{ (predicate types)} ::= \Prop \mid \sty \to \fty.
\end{array}
\]
\fi
Here, \(\INT\) and \(\Prop\) are the types of integers and propositions,
respectively, and \(\sty\to\fty\) is the type of functions
from values of type \(\sty\) to values of type \(\fty\).
Every predicate type \(\fty\) can be written in the form
\(\sty_1\to\cdots\to\sty_n\to\Prop\) (where \(n\ge 0\)),
and thus describes an \(n\)-ary predicate whose arguments have types
\(\sty_1,\ldots,\sty_n\).

For a type \(\sty\), the \emph{order} and \emph{arity} of \(\sty\), written
\(\ord(\sty)\) and \(\ar(\sty)\) respectively, are defined by:
\ifshort
\else
\[
\begin{array}{l}
\ord(\INT)=-1 \qquad \ord(\Prop)=0\\
\ord(\sty\to\fty) = \max(\ord(\fty), \ord(\sty)+1)\\
\ar(\INT)=\ar(\Prop)=0\qquad
\ar(\sty\to\fty) = \ar(\fty)+1.
\end{array}
\]
For example, \(\ord(\INT\to\INT\to\Prop)=0\) and
\(\ord((\INT\to \Prop)\to\Prop)=1\).\footnote{Defining the order
  of \(\INT\) as \(-1\) is a bit unusual, but convenient for stating
  our technical result.}
\fi
  
The set of \emph{\muhflz{} formulas}, ranged over by \(\form\), is given by:
\begin{align*}
&    \form \mbox{ (formulas) } ::= 
    x \mid \form_1\lor \form_2 \mid \form_1\land\form_2 
     \mid \mu x^\fty.\form  
\mid \form_1\form_2 \mid \lambda x^\sty.\form\ 
  \mid \form\, e \mid e_1\le e_2 \\ 
&    e \mbox{ (integer expressions) }::= n \mid x \mid  e_1+e_2 \mid e_1\times e_2.
\end{align*}

We distinguish between formulas and integer expressions.
The metavariables \(x\) and \(n\) range over variables and integer
constants, respectively.
Function application is written as \(\form_1\form_2\) for a predicate
argument and \(\form\,e\) for an integer argument.
A \(\lambda\)-abstraction \(\lambda x^\sty.\form\) denotes a function
that maps an argument \(x\) of type \(\sty\) to the value of \(\form\).
These constructs allow us to manipulate higher-order predicates.
We write \(\TRUE\) and \(\FALSE\) for \(0\le 0\) and \(1\le 0\),
respectively.
The formula \(\mu x^\fty.\form\) denotes the least predicate \(x\)
of type \(\fty\) such that \(x=\form\).
For example, \(\mu x^{\Prop}.x\) is equivalent to \(\FALSE\),
the least solution of \(x=x\).
As another example, the formula
\[
\mu x^{\INT\to\Prop}.\lambda y^{\INT}.(y=0 \lor x(y-1)),
\]
where \(y=0\) abbreviates \(y\le 0\land 0\le y\),
is equivalent to \(\lambda y^{\INT}.0\le y\):
it describes the least set of integers containing \(0\) and
closed under successor.
  
For a formula \(\form\), the \emph{order} of \(\form\) is defined as:
\(\max(\set{0}\cup\set{\ord(\fty)\mid \mu x^\fty.\form'\mbox{ occurs in }\form})\).
We call a \muhflz{} formula \(\form\) \emph{disjunctive} if
the conjunction \(\land\) occurs in \(\form\) only in the form of
\(e_1\le e_2\land \form_1\) (i.e., the left-hand side of each conjunction is
a primitive constraint on integers).

We write \(\seq{\form}_{j,\ldots,k}\) for a sequence of formulas \(\form_j,\ldots,\form_k\);
it denotes an empty sequence if \(k<j\).
We often omit the subscript and just write \(\seq{\form}\) for \(\seq{\form}_{j,\ldots,k}\)
when the subscript
is not important. Similarly, we also write
\(\seq{e}\) and \(\seq{\sty}\) for sequences of expressions and types respectively.
We use the metavariables
\(\fea\), \(\feb\), and \(\fec\) to denote
either a formula or an integer expression.

The simple type system for \muhflz{} formulas is
defined in Figure~\ref{fig:st}, where the judgment
\(\stenv \pST\form:\sty\) means that \(\form\) has type \(\sty\)
under the type environment \(\stenv\), which assigns types to variables.
For example, the following judgment holds:
\[
z\COL\INT\to\Prop \pST
\mu x^{\INT\to\Prop}.\lambda y^{\INT}.(z\,y \lor x(y-1))
:\INT\to\Prop.
\]
This judgment is derived by rule \textsc{T-Mu} from
\[
z\COL\INT\to\Prop,\;x\COL\INT\to\Prop \pST
\lambda y^{\INT}.(z\,y \lor x(y-1)):\INT\to\Prop.
\]
The assumption \(x:\INT\to\Prop\) is used to type the recursive call
\(x(y-1)\).
Rule \textsc{T-Mu} requires the fixpoint variable and the body of
the fixpoint expression to have the same type.

Henceforth, we consider only well-typed formulas (i.e., formulas
\(\form\) such that \(\stenv\pST\form:\sty\) for some \(\stenv\) and \(\sty\)).
A formula \(\form\) is called a \emph{closed} formula of type
\(\sty\) if \(\emptyset\pST\form:\sty\).

\begin{figure*}[tbp]
  \begin{multicols}{2}
    \typicallabel{T-Plus}
  \infrule[T-Var]{}{\stenv,x\COL\sty \pST x:\sty}
  \infrule[T-Or]{\stenv\pST \form_1:\Prop\andalso \stenv\pST \form_2:\Prop}
  {\stenv\pST\form_1\lor\form_2:\Prop}
  \infrule[T-And]{\stenv\pST \form_1:\Prop\andalso \stenv\pST \form_2:\Prop}
  {\stenv\pST\form_1\land\form_2:\Prop}
  \infrule[T-Mu]{\stenv,x\COL\fty\pST \form:\fty}
          {\stenv\pST\mu x^\fty.\form:\fty}
  \infrule[T-App]{\stenv\pST \form_1:\fty_2\to\fty\andalso \stenv\pST\form_2:\fty_2}
          {\stenv\pST\form_1\form_2:\fty}
  \infrule[T-Abs]{\stenv,x\COL\sty\pST\form:\fty}
          {\stenv\pST\lambda x^\sty.\form:\sty\to\fty}
  \infrule[T-AppInt]{\stenv\pST \form:\INT\to\fty\andalso \stenv\pST e:\INT}
          {\stenv\pST\form\,e:\fty}
  \infrule[T-Le]{\stenv\pST e_1:\INT\andalso \stenv\pST e_2:\INT}
          {\stenv\pST e_1\le e_2:\Prop}
  \infrule[T-Int]{}
          {\stenv\pST n:\INT}
  \infrule[T-Plus]{\stenv\pST e_1:\INT\andalso \stenv\pST e_2:\INT}
          {\stenv\pST e_1+ e_2:\INT}
  \infrule[T-Mult]{\stenv\pST e_1:\INT\andalso \stenv\pST e_2:\INT}
          {\stenv\pST e_1\times e_2:\INT}
\end{multicols}
  \caption{Simple Type System for \muhflz{}}
  \label{fig:st}
\end{figure*}

\ifshort
\else
\subsection{Semantics of \muhflz{} Formulas}
For each simple type \(\sty\), we define the partially ordered set
\(\sem{\sty}=(\semd{\sty}, \LEQ_{\sty})\) where \(\LEQ_{\sty}\subseteq
\semd{\sty}\times\semd{\sty}\) by:
\[
\begin{array}{l}
  \semd{\INT}=\Z \quad m\LEQ_{\INT}n \IFF m=n\quad
  \semd{\Prop}=\set{\bot,\top}\quad
  x\LEQ_{\Prop}y \IFF x=\bot\lor y=\top\\
  \semd{\sty\to\fty} = 
  \set{f\in \semd{\sty}\to\semd{\fty}\mid
    \forall x,y\in \semd{\sty}.x\LEQ_{\sty}y \imp f(x)\LEQ_{\fty} f(y)}\\
  f\LEQ_{\sty\to\fty}g \IFF
    \forall x\in \semd{\sty}.f(x)\LEQ_{\fty} g(x).\\
\end{array}
\]
Here, \(\Z\) denotes the set of integers.
For each \(\fty\), \(\sem{\fty}\) (but not \(\sem{\INT}\)) forms a complete lattice.
We write \(\bot_\fty\) (\(\top_\fty\)) for
the least (greatest, resp.) element of \(\sem{\fty}\), and
\(\GLB_\fty\) (\(\LUB_\fty\), resp.) for the greatest lower bound (least upper bound, resp.) operation with respect to \(\LEQ_\fty\).
We also define the least fixpoint operator
\(\LFP_\fty\in \semd{(\fty\to\fty)\to\fty}\) by:
\[
\begin{array}{l}
  \LFP_\fty(f) = \GLB_\fty \set{g\in\semd{\fty}\mid f(g)\LEQ_\fty g}\\
\end{array}
\]

For a simple type environment \(\stenv\), we write \(\semd{\stenv}\)
for the set of maps \(\rho\) such that
\(\dom(\rho)=\dom(\stenv)\) and \(\rho(x)\in\semd{\stenv(x)}\) for each
\(x\in\dom(\rho)\).

For each valid type judgment \(\stenv\pST \form:\sty\), 
its semantics \(\sem{\stenv\pST \form:\sty}\in \semd{\stenv}\to \semd{\sty}\)
is defined by:
\begin{align*}
&  \sem{\stenv,x\COL\sty\pST x\COL\sty}(\rho) = \rho(x)\\
&  \sem{\stenv\pST\form_1\lor\form_2:\Prop}{\rho}
  = 
  \sem{\stenv\pST\form_1:\Prop}{\rho}\LUB_{\Prop}
  \sem{\stenv\pST\form_2:\Prop}{\rho}\\&
  \sem{\stenv\pST\form_1\land\form_2:\Prop}{\rho}
  = 
  \sem{\stenv\pST\form_1:\Prop}{\rho}\GLB_{\Prop}
  \sem{\stenv\pST\form_2:\Prop}{\rho}\\&
  \sem{\stenv\pST \mu x^\fty.\form:\fty}{\rho}
  = 
  \LFP_\fty(\lambda v\in\semd{\fty}.\sem{\stenv,x\COL\fty\pST\form:\fty}(\rho\set{x\mapsto v}))\\&
  \sem{\stenv\pST \lambda x^\sty.\form:\sty\to\fty}{\rho}
  = 
  \lambda w\in\semd{\sty}.\sem{\stenv,x\COL\sty\pST\form:\fty}(\rho\set{x\mapsto w})\\&
  \sem{\stenv\pST \form_1\form_2:\fty}{\rho}
  = 
  \sem{\stenv\pST \form_1:\fty_2\to \fty}{\rho}\,
  (\sem{\stenv\pST \form_2:\fty_2}{\rho})\\&
  \sem{\stenv\pST \form\,e:\fty}{\rho}
  = 
  \sem{\stenv\pST \form:\INT\to \fty}{\rho}\,
  (\sem{\stenv\pST e:\INT}{\rho})\\&
  \sem{\stenv\pST e_1\le e_2:\Prop}{\rho}
  = 
  \left\{\begin{array}{ll}
  \top &\mbox{if $\sem{\stenv\pST e_1:\INT}{\rho}\le\sem{\stenv\pST e_2:\INT}{\rho}$}\\
  \bot &\mbox{otherwise}
  \end{array}\right.\\&
  \sem{\stenv\pST n:\INT}{\rho}=n\\&
  \sem{\stenv\pST e_1+ e_2:\INT}{\rho}
  = 
  \sem{\stenv\pST e_1:\INT}{\rho}+\sem{\stenv\pST e_2:\INT}{\rho}\\&
  \sem{\stenv\pST e_1\times e_2:\INT}{\rho}
  = 
  \sem{\stenv\pST e_1:\INT}{\rho}\times\sem{\stenv\pST e_2:\INT}{\rho}\\
\end{align*}

When the type environment and type are clear from the context,
we write \(\sem{\form}\) for
\(\sem{\stenv\pST\form:\sty}\).
For a closed formula \(\form\), we further omit the empty environment
and write \(\sem{\form}\) for
\(\sem{\emptyset\pST\form:\sty}\emptyset\).
Semantic equality between open formulas of the same type
is understood to hold under all valuations of their free variables.

The \emph{validity checking problem} for \muhflz{} is
the problem of deciding whether
\(\sem{\form}=\top\),
given a closed \muhflz{} formula \(\form\) of type \(\Prop\).
\fi
Note that the validity checking problem  is undecidable.

For closed formulas, the following alternative semantics is sometimes
convenient. Let us define the reduction relation \(\form\red\form'\)
by the following rules.

\begin{multicols}{2}
\infrule{i\in\set{1,2}}{\evC[\form_1\lor\form_2]\red \evC[\form_i]}
\infrule{}{\evC[\TRUE\land \form]\red \evC[\form] }
\infrule{}{\evC[\FALSE\land \form]\red \evC[\FALSE]}
\infrule{}{\evC[\mu x.\form]\red \evC[[\mu x.\form/x]\form]}
\infrule{}{\evC[(\lambda x.\form)e]\red \evC[[e/x]\form]}
\infrule{}{\evC[(\lambda x.\form)\formalt]\red \evC[[\formalt/x]\form]}
\vspace{.23\baselineskip}
\ifshort
\else
\infrule{b = \left\{\begin{array}{ll}
  \TRUE & \mbox{if $\sem{\pST e_1:\INT}\le \sem{\pST e_2:\INT}$}\\
  \FALSE & \mbox{otherwise}
  \end{array}\right.
}{\evC[e_1\le e_2]\red \evC[b]}
\fi
\end{multicols}

Here, \(\evC\) denotes an evaluation context, defined by:
\ifshort
\else
\[
\evC ::= \hole\mid \evC\land \form \mid \evC\,\form \mid \evC\,e.
\]
\fi
We write \(\reds\) for the reflexive and transitive closure of \(\red\).
We have the following fact (see, e.g., \cite{DBLP:conf/lics/Tsukada20}).
\begin{fact}
  \label{fact:red}
Suppose \(\pST \form:\Prop\). Then,
  \(\sem{\form}=\top\) if and only if \(\form\reds \TRUE\).
\end{fact}
Due to the fact above, the validity checking problem is equivalent
to the problem of deciding whether
\(\form \reds\TRUE\), given a closed \muhflz{} formula \(\form\) of type \(\Prop\).

\begin{example}
  \ifshort
\else
  Suppose \(\pST \form:\INT\to\Prop\).
  Then,
  \[ \formalt \DEFEQ (\mu x^{\INT\to\Prop}.\lambda y.\form\,y\lor x(y-1)\lor x(y+1))0 \reds \TRUE\]
  just if \(\form\,n\reds \TRUE\) 
  for some \(n\). Thus,
  the formula \(\formalt\) represents \(\exists z.\form\,z\).
  \fi
\end{example}
\noindent The example above indicates that existential quantifiers on integers
are expressible in \muhflz{}.
Below, we treat existential quantifiers as if they were primitives.

\subsection{Relationship with Reachability Problems}
\label{sec:reachability-as-hfl}

We consider reachability problems for
a call-by-name, simply-typed
\(\lambda\)-calculus extended with two kinds of non-determinism
(\(\dnt\) and \(\ant\)) and a special term \(\err\), which represents
that the designated target has been reached.\footnote{In the context of
  program verification, we are often interested in (un)reachability to bad states.
  Thus, in that context, \(\err\) in this section is actually interpreted as an error state,
and the terms ``angelic'' and ``demonic'' below are swapped.}
The sets of types and terms, ranged over by \(\sigma\) and \(M\) respectively, are
 defined by:
\begin{align*}
\sigma &::= \INT \mid \eta\qquad
\eta ::= \Tunit \mid \sigma\to\eta\\
M &::= \Eunit \mid \err\mid x\mid \lambda x.M \mid M_1\,M_2\mid M\,e\\ &
\mid \Efix{x}{M}{\eta}\mid M_1\dnt M_2\mid M_1\ant M_2\mid \Eass(e_1\le e_2);M.
\end{align*}
Here,
the meta-variable \(e\) ranges over the set of integer expressions
as defined in Section~\ref{sec:syntax}.
The term \(\Efix{x}{M}{\eta}\) denotes a recursive function \(x\) of type \(\eta\) such that \(x=M\).
The term \(M_1\dnt M_2\) denotes a \emph{demonic} choice between \(M_1\) and \(M_2\),
where the choice is up to
the environment (or, the opponent \(\PlayerO\) of the reachability game),
and \(M_1\ant M_2\) denotes an \emph{angelic} choice between \(M_1\) and \(M_2\),
where the choice is up to
the term (or, the player \(\PlayerP\) of the reachability game).
The term \(\Eass(e_1\le e_2);M\) first checks whether \(e_1\le e_2\) holds
and if so, proceeds to evaluate \(M\); otherwise aborts the evaluation of the whole
term. Using \(\Eass\), we can 
express a conditional expression \(\Eif{e_1\le e_2}{M_1}{M_2}\)
as \((\Eass(e_1\le e_2);M_1)\ant (\Eass(e_2+1\le e_1);M_2)\).

A simple type system for the language is given in Figure~\ref{fig:typing-lang}.
In the figure,  \(\toST{\lstenv}\) denotes the type environment obtained by
restricting \(\lstenv\) to bindings on \(\INT\), i.e., \(\toST{\lstenv} =
\set{x\COL\lsty\in \lstenv \mid \lsty=\INT }\).
Henceforth, we consider only well-typed terms.

\begin{figure}
  \begin{multicols}{2}
      \typicallabel{T-Plus}
  \infrule[LT-Unit]{}{\lstenv \pLST \Eunit:\Tunit}
  \infrule[LT-Succ]{}{\lstenv \pLST \err:\Tunit}
  \infrule[LT-Var]{}{\lstenv,x\COL\lsty \pLST x:\lsty}
  \infrule[T-Abs]{\lstenv,x\COL\lsty\pLST M:\lrsty}
          {\lstenv\pLST\lambda x.M:\lsty\to\lrsty}
  \infrule[LT-App]{\lstenv\pLST M_1:\lrsty_2\to\lrsty\andalso \lstenv\pLST M_2:\lrsty_2}
          {\lstenv\pLST M_1 M_2:\lrsty}
  \infrule[LT-AppInt]{\lstenv\pLST M:\INT\to\lrsty\andalso \toST{\lstenv}\pST e:\INT}
          {\lstenv\pLST M\,e:\lrsty}
  \infrule[LT-Fix]{\lstenv,x\COL\lrsty\pLST M:\lrsty}
          {\lstenv\pLST\Efix{x}{M}{\lrsty}:\lrsty}
  \infrule[LT-DChoice]{\lstenv\pLST M_1:\Tunit\andalso \lstenv\pLST M_2:\Tunit}
          {\lstenv\pLST M_1 \dnt M_2:\Tunit}
  \infrule[LT-AChoice]{\lstenv\pLST M_1:\Tunit\andalso \lstenv\pLST M_2:\Tunit}
          {\lstenv\pLST M_1 \ant M_2:\Tunit}
  \infrule[LT-Assume]{\toST{\lstenv}\pST e_1:\INT\andalso \toST{\lstenv}\pST e_2:\INT\\ \lstenv\pLST M:\Tunit}
          {\lstenv\pLST \Eass(e_1\le e_2);M:\Tunit}
\end{multicols}
  \caption{Simple Type System for the Language.}
  \label{fig:typing-lang}
\end{figure}

The order of a type \(\sigma\) is defined by:
\[
\begin{array}{l}
\ord(\INT)=-1 \qquad \ord(\Tunit)=0\qquad
\ord(\sigma\to\eta) = \max(\ord(\eta), \ord(\sigma)+1).
\end{array}
\]
The order of a term \(M\) is defined as the largest order of type \(\eta\)
such that \(M\) has a subterm of the form \(\Efix{x}{M'}{\eta}\).
We write \(\INT^n\to\Prop\) for \(\underbrace{\INT\to\cdots\INT}_n\to\Prop\).

For a closed simply-typed term \(M\) of type \(\Tunit\),
a \emph{play} is a (possibly infinite) sequence of reductions of \(M\).
The play is won by the player \(\PlayerP\) if it ends with \(\err\);
otherwise the play is won by the opponent \(\PlayerO\).
The \emph{reachability game} for \(M\) is the problem of deciding
which player (\(\PlayerP\) or \(\PlayerO\)) 
has a winning strategy.
For the general notion of reachability games and strategies, we refer the reader to
\cite{Automata}.
As a special case of 
the translation of 
Watanabe et al.~\cite{DBLP:conf/pepm/WatanabeTO019}
from temporal properties of programs to \hflz{} formulas,
we obtain the following translation \(\tohflz{(\cdot)}\)
from reachability games to \muhflz{} formulas.
\begin{align*}
  & \tohflz{\Eunit}=\FALSE \quad \tohflz{\err}=\TRUE \quad \tohflz{x}=x \quad 
  \tohflz{(\lambda x.M)}=\lambda x.\tohflz{M} \quad
  \tohflz{(M_1M_2)}=\tohflz{M_1}\tohflz{M_2} \\
  &  \tohflz{(M\,e)}=\tohflz{M}\,e \quad
   \tohflz{(\Efix{x}{M}{})}=\mu x.\tohflz{M}
\quad   \tohflz{(M_1\dnt M_2)}=\tohflz{M_1}\land\tohflz{M_2}
  \\&
  \tohflz{(M_1\ant M_2)}=\tohflz{M_1}\lor\tohflz{M_2} \quad
  \tohflz{(\Eass(e_1\le e_2);M)} = e_1\le e_2\land \tohflz{M}.
\end{align*}

The following is a special case of
the result of Watanabe et al.~\cite{DBLP:conf/pepm/WatanabeTO019}.
\begin{theorem}[\cite{DBLP:conf/pepm/WatanabeTO019}]
  \label{th:reachability}
  For any closed simply-typed term \(M\) of type \(\Tunit\) and order \(k\),
  \(\tohflz{M}\) is a closed \muhflz{} formula of type \(\Prop\) and order \(k\).
The player \(\PlayerP\) wins the reachability game for \(M\), if and only if,
\(\sem{\tohflz{M}}=\top\).
\end{theorem}
Based on the result above, we focus on the validity checking problem
for \muhflz{} formulas, instead of directly discussing the
reachability problem. Note that the may-reachability problem (of asking
whether,
given a closed term \(M\) of which all the branches are angelic,
there exists a reduction sequence from \(M\) to \(\err\))
corresponds to the validity checking problem for disjunctive \muhflz{}
formulas.

\begin{example}
  \label{ex:sum}
  Let us consider the following OCaml program.
  \begin{verbatim}
  let rec sum x k = 
    assert(x>=0); if x=0 then k 0 else sum(x-1)(fun y-> k(x+y))
  in sum n (fun r -> assert(r>=n))
\end{verbatim}
  Suppose we are interested in checking whether the program suffers
  from an assertion failure.
  It is modeled as the reachability problem for the term
  \(M_{\textit{sum}}\,n\,(\lambda r.\Eass(r< n);\err)\),
  where \(M_{\textit{sum}}\) is:
  \begin{align*}
&  \Efix{\textit{sum}}{\lambda x.\lambda k.
      (\Eass(x< 0);\err)
      \\ &\qquad\qquad   \ant
      (\Eass(x=0); k\,0) 
      \ant  (\Eass(x> 0); \textit{sum}\,(x-1)\,(\lambda y.k(x+y)))}{}.
  \end{align*}
  Here, note that an assertion failure is modeled as \(\err\) in our language.
  By Theorem~\ref{th:reachability}, the above term is reachable to \(\err\)
  just if the (disjunctive) \muhflz{} formula \(\form_{\textit{ex1}}
  := \form_{\textit{sum}}\,n\,(\lambda r.r<n)\)
  is valid, where \(\form_{\textit{sum}}\) is:
  \begin{align*}
&  \mu \Sum.\lambda x.\lambda k. 
  x< 0 \lor (x=0\land k\,0)
   \lor (x>0 \land \Sum\,(x-1)\,(\lambda y.k(x+y))).
  \end{align*}
  The formula \(\form_{\textit{ex1}}\) is valid only if \(n< 0\), which implies
  that the OCaml program suffers from an assertion failure just if \(\texttt{n}<0\).
  \qed
\end{example}
\subsection{Main Theorem}
\label{sec:theorem}
For the complexity bounds below, we measure the size of an input by its
syntax, and treat the largest
size of its simple types as a fixed parameter.

The source of the order-decreasing translation is assumed to be
\emph{\(\beta\)-normal}: it has no subformula of the form
\((\lambda x.\formalt)\fea\), including integer applications.
The main theorem is stated under these conventions.
\begin{theorem}
  There exist polynomial-time translations \(\todisj{(\cdot)}\) and
  \(\fromdisj{(\cdot)}\) 
  between order-$n$ \muhflz{} formulas and
  order-$(n+1)$ disjunctive \muhflz{} formulas that satisfy
 the  following properties.
\begin{enumerate}[(i)]
\item  For any order-\(n\) closed \muhflz{} formula \(\form:\Prop\),
  \(\todisj{\form}\) is a closed disjunctive \muhflz{} formula
  of at most order \(n+1\)
  such that \(\sem{\form}=\sem{\todisj{\form}}\).
\item  For any \(\beta\)-normal order-(\(n+1\)) closed disjunctive
  \muhflz{} formula \(\form:\Prop\),
  \(\fromdisj{\form}\) is an order-\(n\) closed \muhflz{} formula
  such that \(\sem{\form}=\sem{\fromdisj{\form}}\).
\end{enumerate}
\end{theorem}

Due to the connection between reachability problems and \muhflz{}
validity checking problems discussed in
Section~\ref{sec:reachability-as-hfl},
the theorem above implies that
under the same type-size convention, any order-\(n\) reachability game
can be converted in polynomial time to an order-(\(n+1\)) may-reachability
problem, and conversely for \(\beta\)-normal source programs.
\ifshort
\else
The result allows us to use a tool for checking the may-reachability of
higher-order programs (such as \mochi{}~\cite{KSU11PLDI}) to solve the reachability game,
and conversely, to use a tool for solving the order-\(n\) reachability game
(such as \nuhflz{} validity
checkers~\cite{DBLP:conf/sas/IwayamaKST20,DBLP:conf/aplas/KatsuraIKT20}
and a HoCHC solver~\cite{DBLP:journals/pacmpl/BurnOR18})
to check the may-reachability of order-(\(n+1\)) programs;
see Section~\ref{sec:app} for more discussion on the applications.
\fi

\section{From Order-$n$ Reachability Games to Order-($n+1$) May-Reachability}
\label{sec:todisj}
In this section, we show the translation \(\todisj{(\cdot)}\) from
order-$n$ \muhflz{} formulas to
order-($n+1$) disjunctive \muhflz{} formulas.
The idea is to transform each proposition \(\form\) (i.e. a formula of type \(\Prop\))
to a predicate \(\todisjaux{\form}\)
of type \(\Prop\to\Prop\), so that \(\TRUE\) and \(\FALSE\) are
respectively converted to
the identity function \(\lambda x.x\) and the constant function
\(\lambda x.\FALSE\). We can then encode
the conjunction \(\form_1\land \form_2\) as
\(\lambda x^\Prop.\todisjaux{\form_1}(\todisjaux{\form_2}x)\),
which is equivalent to the identity function if
both \(\todisjaux{\form_1}\) and \(\todisjaux{\form_2}\) are so,
and is equivalent to \(\lambda x.\FALSE\) if one of
\(\todisjaux{\form_1}\) and \(\todisjaux{\form_2}\) is so.

The translation \(\todisj{(\cdot)}\) for formulas and types is defined as follows.
\begin{align*}
  & \todisj{\form} = \todisjaux{\form}\,\TRUE  \qquad 
  \todisjaux{(e_1\le e_2)} = \lambda x^\Prop.(e_1\le e_2\land x)\qquad \todisjaux{x}=x\\ &
   \todisjaux{(\lambda x^\sty.M)}=\lambda x^{\todisj{\sty}}.\todisjaux{M} \qquad
  \todisjaux{(\form_1\form_2)}=\todisjaux{\form_1}\todisjaux{\form_2} \quad
  \todisjaux{(\form\,e)}=\todisjaux{\form}\,e \quad 
  \todisjaux{(\mu x^\fty.\form)}=\mu x^{\todisj{\fty}}.\todisjaux{\form} 
  \\&
  \todisjaux{(\form_1\lor \form_2)}=\lambda x^\Prop.\todisjaux{\form_1}x\lor \todisjaux{\form_2}x
  \quad
  \todisjaux{(\form_1\land \form_2)}=
  \lambda x^\Prop.\todisjaux{\form_1}(\todisjaux{\form_2}x)\\
  & \todisj{\INT} = \INT\qquad \todisj{\Prop}=\Prop\to\Prop\qquad
  \todisj{(\sty\to\fty)} = \todisj{\sty}\to\todisj{\fty}.
\end{align*}

\begin{example}
  Consider the formula \(\form:=  (\mu p^{\INT\to\Prop}.\lambda y.y=0\lor (p\,(y-1)\land p\,(y+1)))\,n\)
  (where \(n\) is an integer constant).
  \ifshort
\else
  We obtain the following formula as \(\todisj{\form}\):
  \begin{align*}
 &   (\mu p^{\INT\to\Prop\to\Prop}.
    \lambda y.\lambda x^\Prop.(\lambda x^\Prop.y=0\land x)x\\&\qquad\qquad\qquad\lor
    (\lambda x^\Prop.p\,(y-1)\,(p\,(y+1)\,x))\,x)\,n\,\TRUE.
  \end{align*}
  By simplifying the formula with \(\beta\)-reductions, we obtain:
  \begin{align*}
 &   (\mu p^{\INT\to\Prop\to\Prop}.
    \lambda y.\lambda x^\Prop.\\&\qquad
    (y=0\land x)\lor
   p\,(y-1)\,(p\,(y+1)\,x))\,n\,\TRUE.
  \end{align*}
  \fi
\end{example}

The following theorem states the correctness of the translation.
\begin{theorem}
  \label{th:todisj-correctness}
  The translation \(\todisj{(\cdot)}\) is computable in polynomial time.
  If \(\form\) is an order-\(n\) closed \muhflz{} formula of type \(\Prop\),
  then \(\todisj{\form}\) is a closed disjunctive
  \muhflz{} formula of type \(\Prop\) and order at most \(n+1\),
  and \(\sem{\form}=\sem{\todisj{\form}}\).
\end{theorem}

To show the theorem above, we first extend the translation of types
to that of type environments by:
\(\todisj{(x_1\COL\sty_1,\ldots,x_k\COL\sty_k)}
= x_1\COL\todisj{\sty_1},\ldots,x_k\COL\todisj{\sty_k}\).

The following lemma guarantees that the translation preserves typing.
\begin{lemma}
\label{lem:todisj-preserves-types}
If \(\stenv\pST\form:\sty\), then \(\todisj{\stenv}\pST\todisjaux{\form}:\todisj{\sty}\).
\end{lemma}
\begin{proof}
Straightforward induction on the derivation of   \(\stenv\pST\form:\sty\).
\end{proof}  
\begin{corollary}
\label{lem:todisj-is-valid-trans}  
If \(\form\) is an order-\(n\) closed \muhflz{} formula of type \(\Prop\),
then \(\todisj{\form}\) is a closed disjunctive \muhflz{} formula of type \(\Prop\)
and order at most \(n+1\).
\end{corollary}  
\begin{proof}
  Suppose \(\form\) is an order-\(n\) closed \muhflz{} formula.
  By Lemma~\ref{lem:todisj-preserves-types},
  we have \(\emptyset\pST\todisjaux{\form}:\Prop\to\Prop\),
  which implies
  \(\emptyset\pST\todisj{\form}:\Prop\).
  Since each \(\mu\)-formula \(\mu x^\fty.\form'\) in
  \(\form\) is translated to \(\mu x^{\todisj{\fty}}.\todisjaux{\form'}\)
  and \(\ord(\todisj{\fty})=\ord(\fty)+1\),
  \(\todisj{\form}\) is an order-(\(n+1\)) formula. Furthermore,
  all the conjunctions in \(\todisj{\form}\) are of the form \(e_1\le e_2\land \formalt\);
  hence it is disjunctive.
\end{proof}

To prove the latter part of the theorem (i.e., \(\sem{\form}=\sem{\todisj{\form}}\)),
 we define the relation
\(\todisjeq{\sty} \subseteq \sem{\sty}\times \sem{\todisj{\sty}}\) between
the values of the source and the target of the translation, by induction on \(\sty\).
\begin{align*}
  & \todisjeq{\INT} = \set{(n,n) \mid n\in \Z}\\
  & \todisjeq{\Prop} = \set{(\bot, \lambda x\in\sem{\Prop}.\bot)}
  \cup \set{(\top, \lambda x\in\sem{\Prop}.x)}\\
  & \todisjeq{\sty\to\fty} =\\
  &\quad
  \set{(f,g)\mid \forall (v,w)\in\sem{\sty}\times \sem{\todisj{\sty}}.
       v\todisjeq{\sty}w\imp f\,v\todisjeq{\fty}g\,w}.
\end{align*}
We extend \(\todisjeq{\sty}\) pointwise to the relation \(\todisjeq{\stenv}\subseteq
\sem{\stenv}\times \sem{\todisj{\stenv}}\)
on environments by:
\[\rho \todisjeq{\stenv} \rho' \IFF
\rho(x)\todisjeq{\stenv(x)}\rho'(x) \mbox{ for every $x\in\dom(\rho)$}.\]

We first prepare the following lemma.
\begin{lemma}
\label{lem:tr-preserved-by-lfp}  
  If \(f\todisjeq{\fty\to\fty}g\), then
  \(\LFP_{\fty}(f)\todisjeq{\fty}\LFP_{\todisj{\fty}}(g)\).
\end{lemma}
\begin{proof}
  By Cousot and Cousot's fixpoint theorem~\cite{CousotCousot-PJM-82-1-1979},
  there exists an ordinal \(\gamma\) such that
  \(\LFP(f) = f^\gamma(\bot_{\fty})\) and 
  \(\LFP(g) = g^\gamma(\bot_{\todisj{\fty}})\).
  Here, \(f^\gamma(x)\) is defined by:
  \[f^\gamma(x) = \left\{\begin{array}{ll}
    x & \mbox{if $\gamma=0$}\\
f(f^{\gamma'}(x)) & \mbox{if $\gamma=\gamma'+1$}\\
  \LUB_{\gamma'<\gamma} f^{\gamma'}(x) & \mbox{if $\gamma$ is a limit ordinal.}\\
  \end{array}\right.\]
  Thus, it suffices to
  show \(f^\gamma(\bot_{\fty})\todisjeq{\fty}g^\gamma(\bot_{\todisj{\fty}})\)
  by induction on \(\gamma\).
  The case where \(\gamma=0\) or \(\gamma=\gamma'+1\) is trivial.
  Suppose \(\gamma\) is a limit ordinal.
  Suppose \(\fty= \sty_1\to\cdots \to\sty_k\to\Prop\),
  and \(v_i\todisjeq{\sty_i}w_i\) for \(i\in\set{1,\ldots,k}\).
  It suffices to show
  \[f^\gamma(\bot_{\fty})v_1\,\cdots\,v_k
  \todisjeq{\Prop}g^\gamma(\bot_{\todisj{\fty}})w_1\,\cdots\,w_k.\]
  By the induction hypothesis, we have
  \(f^{\gamma'}(\bot_{\fty})\todisjeq{\fty}g^{\gamma'}(\bot_{\todisj{\fty}})\)
  for any \(\gamma'<\gamma\).
  Thus, we have:
  \begin{align*}
    &f^\gamma(\bot_{\fty})v_1\,\cdots\,v_k\\
    &= (\LUB_{\gamma'<\gamma}f^{\gamma'}(\bot_{\fty}))v_1\,\cdots\,v_k\\
    &= \LUB_{\gamma'<\gamma} (f^{\gamma'}(\bot_{\fty})v_1\,\cdots\,v_k)\\
    &\todisjeq{\Prop} \LUB_{\gamma'<\gamma} (g^{\gamma'}(\bot_{\todisj{\fty}})w_1\,\cdots\,w_k)\\
    &= (\LUB_{\gamma'<\gamma} g^{\gamma'}(\bot_{\todisj{\fty}}))w_1\,\cdots\,w_k\\
    &= g^{\gamma}(\bot_{\todisj{\fty}})w_1\,\cdots\,w_k
  \end{align*}
\end{proof}
The semantic preservation property follows from the following lemma.
\begin{lemma}
  \label{lem:todisj}
  Suppose \(\stenv\pST\form:\sty\).
  Then \(\rho\todisjeq{\stenv}\rho'\) implies
  \(\sem{\stenv\pST\form:\sty}\rho \todisjeq{\sty}
  \sem{\todisj{\stenv}\pST\todisjaux{\form}:\todisj{\sty}}\rho'\).
\end{lemma}  
\begin{proof}
  The proof proceeds by induction on the derivation of \(\stenv\pST\form:\sty\).
  Since the other cases are similar or trivial, we show only the main cases.
  \begin{asparaitem}
  \item Case \rn{T-And}: In this case, \(\form=\form_1\land\form_2\)
    and \(\todisjaux{\form}=\lambda x.\todisjaux{\form_1}(\todisjaux{\form_2}\,x)\),
    with \(\sty=\Prop\) and \(\stenv\pST\form_i:\Prop\).
    By the induction hypothesis,
    we have \(\sem{\stenv\pST\form_i:\Prop}\rho\todisjeq{\Prop}
    \sem{\todisj{\stenv}\pST\todisjaux{\form_i}:\todisj{\Prop}}\rho'\)
    for \(i\in\set{1,2}\).
    If \(\sem{\stenv\pST\form:\Prop}\rho=\top\), then
    \(\sem{\stenv\pST\form_i:\Prop}\rho=\top\) for both \(i=1\) and \(2\).
    Thus, \(\sem{\todisj{\stenv}\pST\todisjaux{\form_i}:\todisj{\Prop}}\rho'=\lambda x.x\)
    for both \(i=1\) and \(2\). Therefore,
    we have \(\sem{\todisj{\stenv}\pST\todisjaux{\form}:\todisj{\Prop}}\rho'=\lambda x.x\)
    as required.
    Otherwise, i.e., if \(\sem{\stenv\pST\form:\Prop}\rho=\bot\), then 
    \(\sem{\stenv\pST\form_i:\Prop}\rho=\bot\) for \(i=1\) or \(2\).
    Thus, \(\sem{\todisj{\stenv}\pST\todisjaux{\form_i}:\todisj{\Prop}}\rho'=\lambda x.\bot\)
    for such \(i\).
    Therefore,
    we have \(\sem{\todisj{\stenv}\pST\todisjaux{\form}:\todisj{\Prop}}\rho'=\lambda x.\bot\)
    as required.
  \item Case \rn{T-Mu}: In this case, \(\form=\mu x^{\fty}.\form'\) and
    \(\todisjaux{\form}=\mu x^{\todisj{\fty}}.\todisjaux{\form'}\) with
    \(\sty=\fty\) and  \(\stenv,x\COL\fty\pST \form':\fty\).
    By the induction hypothesis, we have
    \(\sem{\stenv,x\COL\fty\pST\form':\fty}(\rho\set{x\mapsto v})
    \todisjeq{\fty}
    \sem{\todisj{\stenv,x\COL\fty}\pST\todisjaux{\form'}:\todisj{\fty}}(\rho'\set{x\mapsto w})\)
    for any \(v\todisjeq{\fty}w\),
    which implies
    \(\sem{\stenv\pST\lambda x.\form':\fty\to\fty}\rho
    \todisjeq{\fty\to\fty}
    \sem{\todisj{\stenv}\pST\todisjaux{\lambda x.\form'}:\todisj{(\fty\to\fty)}}\rho'\).
  Thus, the required result follows by Lemma~\ref{lem:tr-preserved-by-lfp}.
   \end{asparaitem}
\end{proof}

We are now ready to prove Theorem~\ref{th:todisj-correctness}.
\ifshort
\begin{proof}[Theorem~\ref{th:todisj-correctness}]
\else
\begin{pfof}{Theorem~\ref{th:todisj-correctness}}
  \fi
   Suppose \(\emptyset\pST\form:\Prop\).
   By Corollary~\ref{lem:todisj-is-valid-trans}, \(\todisj{\form}\) is an
   order-\((n+1)\) closed disjunctive \muhflz{} formula.
   By Lemma~\ref{lem:todisj},
   we have \(\sem{\form} = \sem{\emptyset\pST\form:\Prop}\emptyset
   \todisjeq{\Prop}\sem{\emptyset\pST\todisjaux{\form}:\Prop\to\Prop}\emptyset\).
   Thus,
   if \(\sem{\form}=\top\), then
   \(\sem{\todisj{\form}}=\sem{\todisjaux{\form}}\top = (\lambda x.x)\top = \top\).
   If \(\sem{\form}=\bot\), then
   \(\sem{\todisj{\form}}=\sem{\todisjaux{\form}}\top = (\lambda x.\bot)\top = \bot\),
   as required.

The polynomial-time bound follows directly from the definition.
Each translation clause uses the translation of each immediate
subformula only once and introduces only a constant number of
formula constructors.
The type translation also increases the size of each type annotation
by at most a constant factor.
Thus, the translation can be implemented by a single traversal
of the input formula and its type annotations.
\end{pfof}

\section{From Order-($n+1$) May-Reachability to Order-$n$ Reachability Games}
\label{sec:fromdisj}

In this section, we show the translation \(\fromdisj{(\cdot)}\) from
order-($n+1$) disjunctive \muhflz{} formulas to
order-$n$ \muhflz{} formulas.
The translation \(\fromdisj{(\cdot)}\) is much more involved than the
translation \(\todisj{(\cdot)}\).
We first give some intuitions on the first-order case in Section~\ref{sec:order1-tr}.
We then give the translation for the general case in Section~\ref{sec:general-tr},
and prove the correctness in Section~\ref{sec:correctness}.

\subsection{Intuitions for the Order-1 Case}
\label{sec:order1-tr}
Let us recall the formula
\(\form_{\textit{ex1}} := \form_{\textit{sum}}\,n\,(\lambda r.r<n)\)
in Example~\ref{ex:sum},
where \(\form_{\Sum}\COL \INT\to (\INT\to\Prop)\to\Prop\) is:
\begin{align*}
&  \mu \Sum.\lambda x.\lambda k. 
  x< 0 \lor (x=0\land k\,0)
   \lor (x>0 \land \Sum\,(x-1)\,(\lambda y.k(x+y))).
\end{align*}
Note that the order of the formula above is \(1\). We wish to construct
a formula \(\formalt\) of order \(0\), such that
\(\sem{\form_{\Sumex}}=\sem{\formalt}\).
Recall that, by Fact~\ref{fact:red}, \(\sem{\form_{\Sumex}}=\top\) just if \(\form_{\Sumex}\reds\TRUE\).
There are two cases
where the formula \(\form_{\Sumex}\) may be reduced to \(\TRUE\):
\begin{inparaenum}[(i)]
\item \(\form_{\Sumex}\) is reduced to \(\TRUE\) without
  the order-0 argument  \(\lambda r.r<n\) being called; and
\item \(\form_{\Sumex}\) is reduced to \((\lambda r.r<n)m\) for some \(m\),
  and then \((\lambda r.r<n)m\) is reduced to \(\TRUE\).
\end{inparaenum}
Let \(\form_{\Sum_0}\,n\) be the condition for the first case to occur,
and let \(\form_{\Sum_1}\,n\,m\) be the condition that \(\form_{\Sumex}\) is reduced to
\((\lambda r.r<n)m\).
Then, \(\form_{\Sum_0}\) and \(\form_{\Sum_1}\) can be expressed as follows.
\begin{align*}
\form_{\Sum_0} \DEFEQ &  \mu \Sum_0.\lambda x. 
  x< 0
  \lor (x>0 \land \Sum_0\,(x-1)).\\
  \form_{\Sum_1}\DEFEQ & \mu \Sum_1.\lambda x.\lambda z.
  (x=0\land z=0) 
  \lor (x>0\land \exists y.\Sum_1\,(x-1)\,y\land z = x+y).
\end{align*}
To understand the formula \(\form_{\Sum_1}\), notice that
\(\form_{\Sum}\,(x-1)\,(\lambda y.k(x+y))\) is reduced to \(k\,z\) just if
\(\Sum\,(x-1)\,(\lambda y.k(x+y))\) is first reduced to
\((\lambda y.k(x+y))y\) for some \(y\) (the condition for which is expressed by
\({\Sum_1}\,(x-1)\,y\)), and \(z=x+y\) holds.

Using \(\form_{\Sum_0}\) and \(\form_{\Sum_1}\) above,
the formula \(\form_{\textit{ex1}}\) can be translated
\ifshort
\else
to:
\[ \formalt \DEFEQ \form_{\Sum_0}\,n \lor \exists r.\form_{\Sum_1}\,n\,r\land r<n.\]
Note that the order of \(\formalt\) is \(0\).
\fi
In general, if \(\form\) is an order-1 (disjunctive) formula of type
\ifshort
\else
\[ \INT^k \to (\INT^{\ell_1}\to\Prop)\to\cdots \to (\INT^{\ell_m}\to\Prop)\to \Prop\]
\fi
and \(\formalt_i\) (\(i\in\set{1,\ldots,m}\)) is a formula of type \(\INT^{\ell_i}\to\Prop\),
then \(\form\,\seq{e}_{1,\ldots,k}\,\formalt_1\,\cdots\,\formalt_m\)
can be translated to an order-0 formula of the form:
\[
\form_0\, \seq{e}_{1,\ldots,k}\lor \bigvee_{i\in\set{1,\ldots,m}} \exists \seq{y}_{1,\ldots,\ell_i}.
(\form_i\,\seq{e}_{1,\ldots,k}\,\seq{y}_{1,\ldots,\ell_i}\land
 \psi_i\,\seq{y}_{1,\ldots,\ell_i}),
 \]
 where the part \(\form_0\, \seq{e}_{1,\ldots,k}\) expresses the
 condition for \(\form\,\seq{e}_{1,\ldots,k}\,\formalt_1\,\cdots\,\formalt_m\)
 to be reduced to \(\TRUE\) without \(\formalt_i\) being called,
 and the part \(\form_i\,\seq{e}_{1,\ldots,k}\,\seq{y}_{1,\ldots,\ell_i}\)
 expresses the condition for
 \(\form\,\seq{e}_{1,\ldots,k}\,\formalt_1\,\cdots\,\formalt_m\) to be reduced to
 \(\formalt_i\,\seq{y}_{1,\ldots,\ell_i}\).

\subsection{Translation for the General Case}
\label{sec:general-tr}
For higher-order formulas, the translation is more involved.
To simplify the formalization, we assume that 
a formula as an input or output of our translation
is given in the form \((\NT,D, \form_0)\),
called an \emph{equation system}; here
\(D\) is a set of mutually recursive fixpoint equations of the form
\(\set{F_1\,\seq{x}_1 =_\mu \form_1,\ldots, F_n\,\seq{x}_n=_\mu\form_n}\)
and \(\NT\) is the type environment for \(F_1,\ldots,F_n\).
We sometimes omit \(\NT\) and just write \((D, \form_0)\).
Here, each \(\form_i\) (\(i\in\set{0,\ldots,n}\)) should
be fixpoint-free,
\(\form_0\) is well-typed under \(\NT\),
and \(\form_i\) (\(i\in\set{1,\ldots,n}\)) should
have some type \(\fty_i\) 
under the type environment
\( \NT,
x_{i,1}\COL\sty_{i,1},\ldots,x_{i,m_i}\COL\sty_{i,m_i}\),
where \(\NT(F_i) = \sty_{i,1}\to\cdots\to\sty_{i,m_i}\to \fty_i\)
and \(\seq{x}_i = x_{i,1}\cdots x_{i,m_i}\).
The \muhflz{} formula \(\toform{(D, \form_0)}\) represented by
  \((\NT,D, \form_0)\)
is defined by:
\begin{align*}
  & \toform{(\emptyset,\form)} = \form \quad 
  \toform{(D\cup\set{F\,\seq{x}=_\mu \formalt}, \form)}
    = \toform{([\mu F.\lambda \seq{x}.\formalt/F]D, [\mu F.\lambda \seq{x}.\formalt/F]\form)}.
\end{align*}
We write \(\semes{D}{\fm}\) for \(\sem{\toform{(D,\fm)}}\).
This substitution-based definition specifies the semantics; our construction
keeps the equations separate rather than expanding their bodies by substitution.
A polynomial-size encoding into the original formula syntax is outlined
at the end of this subsection.

For an equation system as an input of our translation,
we further assume, without loss of generality,
the following conditions.
\begin{asparaenum}[(I)]
\item 
  Each \(\form_i\) (\(i\in\set{1,\ldots,n}\)) on the right-hand side
  of a definition in \(D\) has type \(\Prop\) and
is generated by the following grammar
(where the metavariable \(x\) may be a fixpoint variable \(F_j\) or its parameters):
\begin{equation}
\label{eq:ESassumption}
  \begin{array}{l}
    \form  ::=
    x \mid \form_1\lor \form_2 \mid e_1\le e_2\land\form
\mid \form_1\form_2 \mid \form\, e.
\end{array}
\end{equation}
  In particular, 
(i) \(\form_i\) is a disjunctive \muhflz{} formula,
(ii) \(\form_i\) contains neither \(\lambda\)-abstractions
  nor fixpoint operators, and (iii)
  a formula of the form \(e_1\le e_2\) may occur only in the form \(e_1\le e_2\land \form\).
  Moreover, every predicate-valued argument is either a variable or a call \(G\,\seq{x}\), where \(G\) is an equation name and all arguments in \(\seq{x}\) are variables.
\item Every integer predicate (i.e., a formula of type of the form \(\INT^\ell\to\Prop\) with \(\ell\ge 0\))
  that occurs in an argument position has the same arity \(\maxar\).
  In other words, in any function type \(\sty\to \tau\),
  either \(\sty=\INT^\maxar\to\Prop\), or \(\ord(\sty)\ne 0\).
\item The ``main formula'' \(\form_0\) is a formula of
  the form \(F\,\lambda \seq{x}_{1,\ldots,M}.\TRUE\).
\end{asparaenum}

Conditions (I)--(III) can be met for the inputs specified in
Section~\ref{sec:theorem}: abstract successful termination as an extra
parameter, pad integer predicates in argument positions to a common arity
\(\maxar\), and lift local least fixpoints and \(\lambda\)-abstractions to
top-level equations. Captured \(\lambda\)-bound variables become extra
parameters; fixpoint variables become global equation names whose calls
supply the captured parameters. Since all fixpoints are least fixpoints,
flattening their nesting preserves the semantics. Moreover, \(\beta\)-normality
ensures that captured variables have order at most \(n\), so lifting
preserves the order-\((n+1)\) bound.

The restriction on predicate-valued arguments in (I)
prevents an exponential blow-up of right-hand sides in the
translation below.
It can be ensured, when necessary, by naming nonvariable
predicate arguments with auxiliary equations. For an argument \(u:\tau\)
with free local parameters \(\seq{x}\), introduce
\(G\,\seq{x}\,\seq{y}=_\mu u\,\seq{y}\), where \(\seq{y}\) saturates
\(\tau\) to \(\Prop\), and replace the argument by \(G\,\seq{x}\).

\ifshort\else
The formula \(\form_{\Sumex}\) given earlier in this section is represented as:
\((\NT_{\Sum},D_{\Sum}, S\,\lambda z.\TRUE)\), where \(D_{\Sum}\) consists of the following equations.
Here, \(\maxar\) is set to \(1\).
\begin{align*}
&  S\; t =_\mu \Sum\,t\,n\,(C\,t)\\
&  C\;t\;r =_\mu r<n\land t\,0\\
  & \Sum\;t\;x\;k =_\mu 
  (x< 0\land t\,0)\lor
  (x= 0\land k\,0)\lor
  (x>0\land \Sum\,t\,(x-1)\,(K\,k\,x))\\
& K\;k\;x\;y =_\mu k(x+y),
\end{align*}
and \(\NT_{\Sum}\) is:
\begin{align*}
&  S\COL (\INT\to\Prop)\to\Prop, C\COL (\INT\to\Prop)\to\INT\to\Prop,\\
  & \Sum\COL (\INT\to\Prop)\to\INT\to(\INT\to\Prop)\to\Prop,\\
  & K\COL (\INT\to\Prop)\to\INT\to\INT\to\Prop.
\end{align*}
\fi

We translate  each equation
\(F\,y_1\,\cdots\,y_m =_\mu  \form\) in \(D\) as follows.
We first decompose the formal parameters \(y_1,\ldots,y_m\) to
two parts: \(y_1,\ldots,y_j\) and \(y_{j+1},\ldots,y_{m}\),
where the orders of (the types of) \(y_{j+1},\ldots,y_m\) are at most \(0\),
and the order of \(y_j\) is at least \(1\); note that
the sequences \(y_1,\ldots,y_j\) and \(y_{j+1},\ldots,y_{m}\) are possibly empty.
We further decompose \(y_{j+1},\ldots,y_{m}\) into order-0 variables
\(x_1,\ldots,x_k\) and integer variables \(z_1,\ldots,z_{p}\) (thus, \(j+k+p=m\)).
Formally, the decomposition of formal parameters is defined 
\ifshort
\else
by:
\begin{align*}
  & \decomparg(\epsilon, \Prop) = (\epsilon, \epsilon, \epsilon)\\
  & \decomparg(u\cdot \seq{y}, \sty\to\fty) =\\
&  
  \left\{
  \begin{array}{ll}
    ((u\COL\sty)\cdot \stenv, \seq{x}, \seq{z}) &\mbox{if 
      \(\decomparg(\seq{y},\fty)=(\stenv, \seq{x}, \seq{z}), \stenv\ne \epsilon\)}\\
    (u\COL \sty, \seq{x}, \seq{z}) &\mbox{if \(\ord(\sty)>0\)}, 
    \mbox{\( \decomparg(\seq{y},\fty)=(\epsilon, \seq{x},\seq{z})\)}\\
    (\epsilon, u\cdot \seq{x}, \seq{z}) &\mbox{if \(\sty=\INT^\maxar\to\Prop\), } 
    \mbox{\(\decomparg(\seq{y},\fty)=(\epsilon, \seq{x},\seq{z})\)}\\
    (\epsilon, \seq{x}, u\cdot\seq{z}) &\mbox{if \(\sty=\INT,
      \decomparg(\seq{y},\fty)=(\epsilon, \seq{x},\seq{z})\)}
  \end{array}
  \right.
\end{align*}
\fi
Here, \(\decomparg(\seq{y}_{1,\ldots,m},\NT(F))\) decomposes the sequence of
variables \(\seq{y}_{1,\ldots,m}\) and returns a triple \((\stenv,\seq{x},\seq{z})\),
where \(\stenv\) is the type environment for \(y_1,\ldots,y_j\),
\(\seq{x}\) is the sequence of integer predicate variables, 
and \(\seq{z}\)
is the sequence of integer variables.

For example, given an equation
\(F\,u_1\,u_2\,u_3\,u_4\,u_5 =_\mu \form\),
where \(\NT(F)=\INT\to((\INT\to\Prop)\to\Prop)\to\INT\to (\INT\to\Prop)\to\INT\to\Prop\),
the formal parameters \(u_1\,\cdots\,u_5\) are decomposed as follows.
\begin{align*}
  &\decomparg(u_1\,\cdots\,u_5, \NT(F)) 
  = (\set{u_1\COL\INT,u_2\COL (\INT\to\Prop)\to\Prop},
u_4, u_3u_5).
\end{align*}

Given an equation 
\(F\,\seq{y} =_\mu  \form\) where
\(\decomparg(\seq{y},\NT(F))=(\stenv,\seq{x}_{1,\ldots,k},\seq{z})\)
with \(\stenv=y_1\COL\sty_1,\ldots,y_j\COL\sty_j\),
we generate equations for new fixpoint variables \(F_0,\ldots,F_k\).
As in the order-1 case,
for \(i\in\set{1,\ldots,k}\),
\(F_i\,\seq{\form}'_{1,\ldots,j}\,\seq{z}\,\seq{u}_{1,\ldots,\maxar}\)
  represents the condition for
  \(F\,\seq{\form}_{1,\ldots,j}\,\seq{w}\)
  to be reduced to \(x_i\,\seq{u}_{1,\ldots,\maxar}\)
(where \(\seq{\form}'_{1,\ldots,j}\) is the sequence of formulas obtained by
  translating \(\seq{\form}_{1,\ldots,j}\) in a recursive manner,
  and \(\seq{w}\) is a sequence obtained by shuffling \(\seq{x}_{1,\ldots,k}\) and \(\seq{z}\)).
\(F_0\) is a new component required to deal with higher-order formulas;
it is used to compute the condition for
\(F\,\seq{y}\) to be reduced to \(x\,\seq{u}_{1,\ldots,\ell_i}\) for some order-0
predicate \(x\), which has been passed through higher-order parameters
\(\seq{y}_{1,\ldots,j}\).
For example, consider a formula
\(F\,(G\,x)\,y\) where
\(F\COL ((\INT\to\Prop)\to\Prop)\to (\INT\to\Prop)\to\Prop,
G\COL (\INT\to\Prop)\to(\INT\to\Prop)\to\Prop\).
Then, the condition for \(F\,(G\,x)\,y\) to be reduced to \(y\,n\) is
computed by using \(F_1\), while
the condition for \(F\,(G\,x)\,y\) to be reduced to \(x\,n\) is
computed by using \(F_0\); see Example~\ref{ex:f0}
for a concrete version of
this example.

To compute \(F_0,\ldots,F_k\), we translate each subformula \(\form\) of
the body of \(F\) to:
\ifshort
\else
\[(\form_\Mark, \form_0,\form_1,\ldots,\form_k,\form_{k+1},\ldots,\form_{k+\gar(\fty)}),\]
\fi
where \(\fty\) is the type of \(\form\),
and \(\gar(\fty)\) denotes the number of order-0 arguments passed after
the last argument of order greater than \(0\). More precisely,
we define the decomposition of types as follows.
\begin{align*}
  & \decomp(\Prop) = (\epsilon, 0, 0)\\
  & \decomp(\sty\to\fty) = 
  \left\{
  \begin{array}{ll}
    (\sty\cdot \seq{\sty}, m, n) &\mbox{if 
      \(\decomp(\fty)=(\seq{\sty}, m,n), \seq{\sty}\ne \epsilon\)}\\
    (\sty, m, n) &\mbox{if \(\ord(\sty)>0,
      \decomp(\fty)=(\epsilon, m,n)\)}\\
    (\epsilon, m+1, n) &\mbox{if \(\sty=\INT^\maxar\to\Prop,
      \decomp(\fty)=(\epsilon, m,n)\)}\\
    (\epsilon, m, n+1) &\mbox{if \(\sty=\INT,
      \decomp(\fty)=(\epsilon, m,n)\)}
  \end{array}
  \right.
\end{align*}
Then, \(\gar(\fty)\) denotes \(m\)
when \(\decomp(\fty)=(\seq{\sty},m,n)\).
For example, for \(\fty=(\INT\to\Prop)\to
((\INT\to\Prop)\to \Prop)\to (\INT\to\Prop)\to \INT\to(\INT\to\Prop)\to\Prop\),
\(\decomp(\fty)=((\INT\to\Prop)\cdot ((\INT\to\Prop)\to\Prop), 2, 1)\);
hence \(\gar(\fty)=2\).
Here, \(\form_1,\ldots,\form_k\) are analogous to \(F_1,\ldots,F_k\): they
are used for computing the condition for \(\form\,\seq{\psi}\) to be reduced
to \(x_i\,\seq{n}\). Similarly, \(\form_{k+i}\) (where \(i\in\set{1,\ldots,\gar(\fty)}\))
is used for computing the condition for \(\form\,\seq{\psi}\) to be reduced
to \(\psi_i\,\seq{n}\), where \(\psi_i\) is the \(i\)-th order-0 argument of \(\form\).
The component \(\form_0\) is analogous to \(F_0\), and used to compute the condition for
\(\form\,\seq{\psi}\) to be reduced to \(x\,\seq{n}\), where
\(x\) is an order-0 predicate passed through higher-order arguments of \(\form\).
The other component \(\form_\Mark\) is similar to \(\form_0\), but the target
\ifshort \else order-0 \fi predicate \(x\) may have already been set inside \(\form_\Mark\).

Based on the intuition above, we formalize the translation of
a formula as the relation:
\ifshort
\else
\[
\stenv;\penv \pn \form:\fty \tr (\form_\Mark,\form_0,\ldots,\form_{k+\gar(\fty)}).
\]
\fi
Here, \(\NT\) denotes the type environment for fixpoint variables defined by \(D\).
If \(\form\) is a subformula of the body of \(F\),
and \(F\) is defined by \(F\,\seq{y} =_\mu \form_F\),
then 
\(\stenv\) and \(\penv\)
 are set to \(\stenv_F,\seq{z}\COL\seq{\INT}\) and \(\seq{x}_F\) respectively,
 where \(\decomparg(\seq{y},\NT(F))=(\stenv_F,\seq{x}_F,\seq{z})\).

 \ifshort

 \else
 The output \((\form_\Mark,\form_0,\ldots,\form_{k+\gar(\fty)})\) of the translation
 has type
 \(\fromdisjT{\fty}{k+2}\)
 under the type environment \(\fromdisjTEp{\NT},
 \fromdisjTE{\stenv}\),
 where the translations of types and type environments are defined by:
\begin{align*}
&  \fromdisjT{\INT}{k} = \INT\\
&  \fromdisjT{\tau}{k} 
  =
  (\Pi_{i=1,\ldots,k} (\fromdisjT{\seq{\sty}}{2}\to \INT^{n+\maxar}\to\Prop)) 
  \times 
  (\Pi_{i=1,\ldots,m} (\fromdisjT{\seq{\sty}}{1}\to \INT^{n+\maxar}\to\Prop))
  \\  & \qquad\qquad\qquad\qquad\qquad\qquad\qquad\qquad\qquad\qquad
  \mbox{(if $\decomp(\tau)=(\seq{\sty}, m, n)$)}\\
  & \fromdisjTE{\emptyset} = \emptyset\\
  & \fromdisjTE{(\stenv, y\COL\INT)} =
  \fromdisjTE{\stenv},y\COL\INT\\
  & \fromdisjTE{(\stenv, y\COL\fty)} =
  \fromdisjTE{\stenv},y_\Mark\COL\fty_\Mark, y_0\COL\fty_0,\ldots,y_k\COL\fty_k
  \mbox{where $\fromdisjT{\fty}{2}=
    \fty_\Mark\times \fty_0\times \cdots \times
    \fty_k$}\\
  & \fromdisjTEp{\emptyset} = \emptyset\\
  & \fromdisjTEp{(\NT, F\COL\fty)} =
  \fromdisjTEp{\NT},F_0\COL\fty_0,\ldots,F_k\COL\fty_k 
  \mbox{where $\fromdisjT{\fty}{1}=
   \fty_0\times \cdots \times \fty_k$}.
\end{align*}
Here, we have extended simple types with product types; we extend the definition of
the order of a type by: \(\ord(\fty_1\times \cdots \times \fty_n)
=\max(\ord(\fty_1),\ldots,\ord(\fty_n))\).
Note that the translation of a type decreases
the order of the type by one, i.e., \(\ord(\fromdisjT{\fty}{k}) = \max(0, \ord(\fty)-1)\).
\fi

\begin{figure*}[tbp]
  \typicallabel{Tr-AppG}
  \infrule[Tr-VarG]{
    \form_j =\left\{\begin{array}{ll}
      \lambda \seq{z}_{1,\ldots,\maxar}.\lambda \seq{w}_{1,\ldots,\maxar}.  \wedge_{p=1,\ldots,\maxar}(z_p=w_p),
       &\mbox{if $j=i$}\\
      \lambda \seq{z}_{1,\ldots,\maxar}.\lambda \seq{w}_{1,\ldots,\maxar}.\FALSE
      & \mbox{otherwise}
    \end{array}\right.
  }{\stenv;
    \penv
\pn x_i:\INT^{\maxar}\to\Prop\tr
  (\form_\Mark,\form_0,\ldots, \form_{k})}

  \infrule[Tr-Var]{
    \decomp(\stenv(y))=(\seq{\sty}, m, p)\\
  }
        {\stenv; \penv
          \pn y: \stenv(y) \tr
          (y_\Mark, \underbrace{y_0,\ldots,y_0}_{k+1}, y_1,\ldots,y_m)}

        \infrule[Tr-VarF]{
          \decomp(\NT(F))=(\seq{\sty},m, p)
  }
        {\stenv; \penv
           \pn F:\NT(F)\tr
          (F_0, \underbrace{F_0, \ldots,F_0}_{k+1}, F_1,\ldots,F_m)}

  \infrule[Tr-Le]
          {\stenv; \penv \pn \form:\Prop \tr (\form_\Mark, \form_0,\ldots,\form_{k})\andalso
            \formalt_j = \lambda \seq{z}_{1,\ldots,\maxar}.
            \bra{e_1\le e_2}{\form_j\,\seq{z}_{1,\ldots,\maxar}}
          }
        {\stenv;\penv
 \pn e_1\le e_2\land \form:\Prop \tr
 (\formalt_\Mark,\formalt_0,\ldots,\formalt_k)}

\ifnarrow
  \infrule[Tr-App]{
 \ord(\sty_0\to\fty)>1\andalso   \gar(\sty_0\to\fty)=m
    \andalso
\gar(\sty_0) = m'\\
    \stenv;\penv
    \pn \form:\sty_0\to \fty
 \tr (\form_\Mark,\form_0,\ldots,\form_{k+m})\andalso
    \stenv;\penv \pn \formalt:\sty_0\tr
    (\formalt_\Mark,\formalt_0,\ldots,\formalt_{k+m'})
  }
       {\stenv;\penv
            \pn \form\,\formalt:\fty 
\tr (\form_\Mark(\formalt_\Mark,\formalt_0,\formalt_{k+1},\ldots,\formalt_{k+m'}),
\form_0(\formalt_0,\formalt_0,\formalt_{k+1},\ldots,\formalt_{k+m'}),\\\qquad\qquad\qquad
\form_1(\formalt_1,\formalt_0,\formalt_{k+1},\ldots,\formalt_{k+m'}),\ldots,
\form_k(\formalt_k,\formalt_0,\formalt_{k+1},\ldots,\formalt_{k+m'}),\\\qquad\qquad\qquad
\form_{k+1}(\formalt_0,\formalt_{k+1},\ldots,\formalt_{k+m'}),
\ldots,
\form_{k+m}(\formalt_0,\formalt_{k+1},\ldots,\formalt_{k+m'}))}
\else
  \infrule[Tr-App]{
 \ord(\sty_0\to\fty)>1\andalso   \gar(\sty_0\to\fty)=m
    \andalso
\gar(\sty_0) = m'\andalso
    \stenv;\penv
    \pn \form:\sty_0\to \fty
 \tr (\form_\Mark,\form_0,\ldots,\form_{k+m})\\
    \stenv;\penv \pn \formalt:\sty_0\tr
    (\formalt_\Mark,\formalt_0,\ldots,\formalt_{k+m'})
  }
       {\stenv;\penv
            \pn \form\,\formalt:\fty 
\tr (\form_\Mark(\formalt_\Mark,\formalt_0,\formalt_{k+1},\ldots,\formalt_{k+m'}),
\form_0(\formalt_0,\formalt_0,\formalt_{k+1},\ldots,\formalt_{k+m'}),
\form_1(\formalt_1,\formalt_0,\formalt_{k+1},\ldots,\formalt_{k+m'}),\ldots,\\\qquad\qquad\qquad
\form_k(\formalt_k,\formalt_0,\formalt_{k+1},\ldots,\formalt_{k+m'}),
\form_{k+1}(\formalt_0,\formalt_{k+1},\ldots,\formalt_{k+m'}),
\ldots,
\form_{k+m}(\formalt_0,\formalt_{k+1},\ldots,\formalt_{k+m'}))}
\fi       

  \infrule[Tr-AppG]{
\decomp(\fty)=(\epsilon, m-1,p)\andalso
    \stenv;\penv
    \pn \form:(\INT^{\maxar}\to\Prop)\to \fty\tr (\form_\Mark,\form_0,\ldots,\form_{k+m})\\
    \stenv;  \penv
    \pn \formalt:{\INT}^{\maxar}\to\Prop\tr (\formalt_\Mark,\formalt_0,\ldots,\formalt_{k})\\
    \formaltp_j = 
    \lambda \seq{z}_{1,\ldots,p}.\lambda \seq{w}_{1,\ldots,\maxar}.
    \form_j\,\seq{z}\,\seq{w} \lor \exists \seq{u}_{1,\ldots,\maxar}.
    (\form_{k+1}\,\seq{z}\,\seq{u}_{1,\ldots,\maxar}
    \land \formalt_j\,\seq{u}_{1,\ldots,\maxar}\,\seq{w}_{1,\ldots,\maxar})
  }
          {\stenv;\penv
            \pn \form\,\formalt: \fty\tr
            (\formaltp_\Mark, \formaltp_0,\ldots,\formaltp_k,
            \form_{k+2},\ldots,\form_{k+m})}

  \infrule[Tr-AppI]{\stenv; \penv
    \pn \form:\INT\to \fty\tr (\form_\Mark,\form_0,\ldots,\form_{\ell+k})}
          {\stenv;\penv
            \pn \form\,e:\fty \tr
(\form_\Mark\,e, \form_0\,e,\ldots,\form_{\ell+k}\,e)}

\infrule[Tr-Disj]
        {\stenv;\penv
           \pn \form:\Prop\tr (\form_\Mark,\form_0,\ldots,\form_{k})\andalso
          \stenv;\penv
          \pn \formalt:\Prop\tr (\formalt_\Mark,\formalt_0,\ldots,\formalt_{k})\\
          \formaltp_j = \lambda \seq{z}_{1,\ldots,\maxar}.
                   {\form_j\,\seq{z}_{1,\ldots,\maxar}}
                   \lor{\formalt_j\,\seq{z}_{1,\ldots,\maxar}}
                   }
        {\stenv;\penv
          \pn \form \nc \formalt:\Prop\tr (\formaltp_\Mark,
          \formaltp_0,
          \ldots,\formaltp_{k})}
        \caption{Translation from order-(\(n+1\)) disjunctive \muhflz{} to
          order-\(n\) \muhflz{}. }
\label{fig:tr-n-1}
\end{figure*}

\begin{figure*}[tbp]
  \infrule[Tr-Def]
          {\decomparg(\seq{w},\NT(F)) =
            (\seq{y}\COL\seq{\sty},
            \penv,
            \seq{z})\\
            y_1\COL\sty_1,\ldots,y_m\COL\sty_m,\seq{z}\COL\seq{\INT};
            \penv
    \pn \form:\Prop\tr (\form_\Mark,\form_0,\ldots,\form_k)\\
\Bigg\{\begin{aligned}\
& \seq{y}_i=(y_{i,\Mark},y_{i,0},\ldots,y_{i,\gar(\sty_i)})
&&\seq{y}'_i=(y_{i,0},\ldots,y_{i,\gar(\sty_i)})
&&\text{if }i \in \set{1,\dots,m},\ \sty_i \neq \INT
\\
& \seq{y}_i=y_i
&&\seq{y}'_i=y_i
&&\text{if }i \in \set{1,\dots,m},\ \sty_i = \INT
\end{aligned}
}
  {\pn (F\,\seq{w}=_\mu 
    \form)\tr
    \set{F_0\,\seq{y}_1\,\cdots\,\seq{y}_m\,\seq{z}=_\mu \form_\Mark}
    \cup 
    \set{F_i\,\seq{y}'_1\,\cdots\,\seq{y}'_m\,\seq{z}=_\mu \form_i\mid
     i\in\set{1,\ldots,k}}
  }
  
  \infrule[Tr-Main]
          {D' = \bigcup\set{D''\mid \pn (F\,\seq{y}=_\mu \form)\tr D''\mid
              F\,\seq{y}=_\mu \form\in D}}
          {(D, S\,\lambda \seq{z}.\TRUE)\tr
             (D', \exists \seq{z}.S_1\,\seq{z})}

        \caption{Translation from order-(\(n+1\)) disjunctive \muhflz{} to
          order-\(n\) \muhflz{}, Part II. }
\label{fig:tr-n-1:II}
\end{figure*}

The translation rules are given in Figures~\ref{fig:tr-n-1} and \ref{fig:tr-n-1:II}.
We explain the main rules below.
In the rule \rn{Tr-VarG} for an order-0 variable \(x_i\)
(which should disappear after the translation), 
\(\form_j\;\seq{z}_{1,\ldots,\maxar}\;\seq{w}_{1,\ldots,\maxar}\) 
should represent the condition for 
\(x_i\;\seq{z}_{1,\ldots,\maxar}\reds x_j\;\seq{w}_{1,\ldots,\maxar}\);
thus \(\form_j\) is defined so that
\(\form_j\;\seq{z}_{1,\ldots,\maxar}\;\seq{w}_{1,\ldots,\maxar}\)
is equivalent to \(\TRUE\) just if \(i=j\) and \(\seq{z}_{1,\ldots,\maxar}=
\seq{w}_{1,\ldots,\maxar}\).
In the rule \rn{Tr-Var} for a variable \(y\) in \(\stenv\),
the output of the translation is constructed from \((y_\Mark,y_0,y_1,\ldots,y_m)\),
whose values will be provided by the environment.
Because the environment does not know order-0 variables \(x_1,\ldots,x_k\),
we use \(y_0\) to compute the condition for 
\(y\,\seq{\formalt}\) to be reduced
to \(x_i\,\seq{m}\).
The rule \rn{Tr-VarF} for fixpoint variables is almost the same as \rn{Tr-Var},
except that the component \(F_0\) is reused for \(F_\Mark\). The rationale for this
is as follows: both \(\form_{\Mark}\) and \(\form_0\) are used for computing
the condition for a target order-0 predicate variable (which is set by the environment)
to be reached, and the only difference between them is that the target predicate
may have already been set in \(\form_{\Mark}\), but since \(F\) is a closed formula,
such distinction does not make any difference; hence \(F_0\) and \(F_\Mark\) need not be
distinguished from each other.

In the rule \rn{Tr-App}, the first two
components (\(\form_\Mark(\formalt_\Mark,\ldots)\) and
\(\form_0(\formalt_0,\ldots)\)) are used for computing the
condition for some target predicates (set by the environment)
to be reached, and the next \(k\) components
(\(\form_1(\formalt_1,\ldots),\ldots,
\form_k(\formalt_k,\ldots)\)) are used for computing the
condition for predicate \(x_1,\ldots,x_k\) to be reached.
The rule \rn{Tr-AppG} is another rule for applications,
where the argument \(\formalt\) is an order-0 predicate.
The component \(\formaltp_j\) of the output is used for computing
the condition for the predicate \(x_j\) to be reached
(i.e., the condition for a formula of the form
\(\form\,\formalt\,\seq{\formalt'}\)
to be reduced to 
\(x_j\,\seq{w}_{1,\ldots,\maxar}\),
where \(\seq{\formalt'}\) consists of order-0 predicates
and integer arguments \(\seq{z}_{1,\ldots,p}\)).
The formula \(\form\,\formalt\,\seq{\formalt'}\)
may be reduced to
\(x_j\,\seq{w}_{1,\ldots,\maxar}\) if either (i)
\(\form\,\formalt\,\seq{\formalt'}\reds
x_j\,\seq{w}_{1,\ldots,\maxar}\) without \(\formalt\) being called,
or (ii)
\(\form\,\formalt\,\seq{\formalt'}\) is reduced to
\(\formalt\,\seq{u}\) for some \(\seq{u}\),
and \ifshort\else then \fi \(\formalt\,\seq{u}\) is reduced to
\(x_j\,\seq{w}_{1,\ldots,\maxar}\).
The part
\(\form_j\,\seq{z}\,\seq{w}\) represents the former condition,
and the part \(\exists \seq{u}.\cdots\) represents
\ifshort\else the latter condition. \fi
\ifshort\else In the rule \rn{Tr-Def} for definitions,
the bodies of the definitions for \(F_0,\ldots,F_k\)
are set to the corresponding components of the translation of
the body of \(F\). \fi

\begin{example}
\label{ex:f0}
Consider \(S\,(\lambda x.\TRUE)\), where \(S\) is defined by:
\ifshort
\else
\begin{align*}
  & S\,t =_\mu F\,(G\,t)\,t\\
  & F\,v\,w =_\mu v\,H\lor w\,2\\
  & G\,p\,q =_\mu p\,1\\
  & H\,x =_\mu H\,x
\end{align*}
\fi
There are the following two ways for \(S\,t\) to be reduced to \(t\,n\) for some \(n\):
\begin{align*}
  &  S\,t \red F\,(G\,t)\,t \red G\,t\,H \lor t\,2 \red G\,t\,H\red t\,1\\&
  S\,t \red F\,(G\,t)\,t \red G\,t\,H \lor t\,2 \red t\,2.
\end{align*}
The output of our transformations (with some simplification) is \(\exists z.S_1\,z\) where:
\begin{align*}
  & S_1 =_\mu
 \lambda w_1.F_0\,(\lambda w_1.G_0\,w_1\lor G_1\,w_1, G_0, G_2)\,w_1
  \lor F_1\,(G_0,G_2)\,w_1\\ &
  F_0\,(v_\Mark,v_0,v_1) =_\mu \lambda z_1.v_\Mark\,z_1\lor \exists u_1.v_1\,u_1\land H_0\,u_1\,z_1\\&
  F_1\,(v_0,v_1) =_\mu \lambda z_1. v_0\,z_1\lor (\exists u_1.v_1\,u_1\land H_0\,u_1\,z_1)\lor 2=z_1\\&
  G_0 =_\mu \lambda w_1.\FALSE \quad 
  G_1 =_\mu \lambda w_1.1=w_1 \quad 
  G_2 =_\mu \lambda w_1.\FALSE\quad 
  H_0\,x=_\mu H_0\,x.
\end{align*}
Notice that the formula \(S_1\,z\) has the following two reduction sequences that lead to the
conditions of the form \(z=n\) for some \(n\).
{\small
\begin{align*}
  &  S_1\,z \reds F_0\,(\lambda w_1.G_0\,w_1\lor G_1\,w_1, G_0, G_2)\,z 
  \reds (\lambda w_1.G_0\,w_1\lor G_1\,w_1)z 
  \reds 1=z\\
  &  S_1\,z \reds F_1\,(G_0,G_2)\,z 
   \reds G_0\,z\lor (\exists u_1.G_2\,u_1\land H_0\,u_1\,z)\lor 2=z 
   \reds 2=z.
\end{align*}
}
The former reduction sequence corresponds to
the reduction sequence of the original formula \(S\,t \reds t\,1\) where \(t\) embedded in
the first argument of \(F\) (in \(F\,(G\,t)\,t\)) is called,
and the latter reduction sequence corresponds to 
the reduction sequence \(S\,t \reds t\,2\) where the second argument \(t\) of \(F\) (in
\(F\,(G\,t)\,t\)) is called. Note that the first condition \(1=z\) has been computed by using
\(F_0\), and the second condition \(2=z\) has been computed by using \(F_1\). \qed
\end{example}

\ifshort\else
\begin{example}
  \ifshort
  The formula \(\form_{\Sumex}\) given in Example~\ref{ex:sum} is represented as:
\((\NT_{\Sum},D_{\Sum}, S\,\lambda z.\TRUE)\), where \(D_{\Sum}\) consists of the following equations.
Here, \(\maxar\) is set to \(1\).
\begin{align*}
&  S\; t =_\mu \Sum\,t\,n\,(C\,t)\\
&  C\;t\;r =_\mu r<n\land t\,0\\
  & \Sum\;t\;x\;k =_\mu\\&
  (x< 0\land t\,0)\lor
  (x= 0\land k\,0)\lor
  (x>0\land \Sum\,t\,(x-1)\,(K\,k\,x))\\
& K\;k\;x\;y =_\mu k(x+y),
\end{align*}
and \(\NT_{\Sum}\) is:
\begin{align*}
&  S\COL (\INT\to\Prop)\to\Prop, C\COL (\INT\to\Prop)\to\INT\to\Prop,\\
  & \Sum\COL (\INT\to\Prop)\to\INT\to(\INT\to\Prop)\to\Prop,\\
  & K\COL (\INT\to\Prop)\to\INT\to\INT\to\Prop.
\end{align*}

  \else
  Recall the example of \(D_{\Sum}\) given earlier in this section.
  \fi
  The following is the output of the translation (with some simplification
  by \(\beta\)-reductions and simple quantifier eliminations).
  \begin{align*}
    & S_0 =_\mu \lambda w_1.\Sum_0\,n\,w_1 \lor \exists u_1.\Sum_2\,n\,u_1\land C_0\,u_1\,w_1\\
    & S_1 =_\mu \lambda w_1.\Sum_0\,n\,w_1\lor \Sum_1\,n\,w_1 
    \lor \exists u_1.\Sum_2\,n\,u_1\land
    (C_0\,u_1\,w_1\lor C_1\,u_1\,w_1)\\
    & C_0\,x =_\mu \lambda z_1.\FALSE\\
    & C_1\,x =_\mu \lambda z_1.x<n\land 0=z_1\\
    & \Sum_0\,x =_\mu \lambda z_1.(x>0\land 
    (\Sum_0\,(x-1)\,z_1 
    \lor \exists u_1.\Sum_2\,(x-1)\,u_1\land K_0\,x\,u_1\,z_1))\\
& \Sum_1\,x =_\mu \lambda z_1.
    (x<0\land 0=z_1)\\
& \qquad\lor
    (x>0\land
      (\Sum_0\,(x-1)\,z_1
       \lor \Sum_1\,(x-1)\,z_1\\
& \qquad\qquad\lor
       \exists u_1.\Sum_2\,(x-1)\,u_1
          \land K_0\,x\,u_1\,z_1)).    \\
    & \Sum_2\,x =_\mu \lambda z_1.x=0\land 0=z_1\\&\qquad\qquad\qquad
    \lor (x>0\land
    (\Sum_0\,(x-1)\,z_1 \\&\qquad\qquad\qquad\qquad
    \lor \exists u_1.\Sum_2\,(x-1)\,u_1
    \land (K_0\,x\,u_1\,z_1
    \lor \exists u_2. (K_1\,x\,u_1\,u_2\land u_2=z_1))))\\
    & K_0\,x\,y =_\mu \lambda w_1.\FALSE\\
    & K_1\,x\,y =_\mu \lambda w_1.x+y=w_1.
  \end{align*}
  Although the output may look complicated, since the order of the resulting formula is
  \(0\), we can directly translate its validity checking problem
  to a CHC solving problem using the method of \cite{DBLP:conf/sas/0001NIU19},
  for which various automated solvers are available~\cite{DBLP:journals/fmsd/KomuravelliGC16,Eldarica,DBLP:journals/jar/ChampionCKS20}.
  \qed
\end{example}

\begin{example}
  Let us consider the formula \(S\,\lambda z.\TRUE\)\footnote{Taken from \cite{DBLP:conf/sas/IwayamaKST20}.}, where:
  \newcommand\Plus{\mathit{plus}}
  \begin{align*}
  &  S\,t =_\mu \Sum\,\Plus\,n\,(C\,t)\\&
    C\,t\,x =_\mu x<n\land t\,0\\&
    \Sum\,f\,x\,k =_\mu x\le 0\land k\,0 \lor x>0\land f\,x\,(D\,f\,x\,k)\\&
    \Plus\,x\,k =_\mu k(x+x)\\&
    D\,f\,x\,k\,y =_\mu \Sum\,f\,(x-1)\,(E\,y\,k)\\&
    E\,y\,k\,z =_\mu k(y+z).
  \end{align*}
  It is translated to \(\exists z.S_1\,z\), where:\footnote{This is a mechanically generated output
    based on the transformation rules,
  followed by slight manual simplification.}
  \begin{align*}
    & S_1 =_\mu \lambda w_1.\Sum_0\,(\Plus_0, \Plus_0, \Plus_1)\,n\,w_1\\& \quad
    \lor \exists u_1.\Sum_1(\Plus_0,\Plus_1)\,n\,u_1\land (C_0\,u_1\,w_1\lor C_1\,u_1\,w_1)\\&
    C_0\,x =_\mu \lambda z_1.\FALSE\\&
    C_1\,x =_\mu \lambda w_1.x<n\land 0=w_1\\&
    \Sum_0\,(f_\Mark,f_0,f_1)\,x =_\mu
    \lambda z_1.x>0\land (f_\Mark\,x\,z_1\\&\qquad\qquad\lor
    \exists u_1.f_1\,x\,u_1\land D_0(f_\Mark,f_0,f_1)\,x\,u_1\,z_1)\\&
    \Sum_1\,(f_0,f_1)\,x =_\mu \lambda z_1.x\le 0\land 0=z_1 
    \lor
    x>0\land (f_0\,x\,z_1\lor \\&\qquad \exists u_1.
    f_1\,x\,u_1\land (D_0(f_0,f_0,f_1)\,x\,u_1\,z_1
    \lor D_1(f_0,f_1)\,x\,u_1\,z_1))\\&
    \Plus_0\,x =_\mu \lambda w_1.\FALSE\\&
    \Plus_1\,x =_\mu \lambda w_1.x+x=w_1\\&
    D_0\,(f_\Mark,f_0,f_1)\,x\,y =_\mu \lambda w_1.\Sum_0\,(f_\Mark,f_0,f_1)\,(x-1)\,w_1\\&\qquad\lor
    \exists u_1.\Sum_1\,(f_0,f_1)\,(x-1)\,u_1\land E_0\,y\,u_1\,w_1\\&
    D_1\,(f_0,f_1)\,x\,y =_\mu \lambda w_1.\Sum_0\,(f_0,f_0,f_1)\,(x-1)\,w_1\\&\qquad
    \lor
    \exists u_1.\Sum_1(f_0,f_1)\,u_1\\&\qquad\qquad\land (E_0\,y\,u_1\,w_1\lor \exists u_2.E_1\,y\,u_1\,u_2\land u_2=w_1)\\&
    E_0\,y\,z =_\mu \lambda w_1.\FALSE\\&
    E_1\,y\,z =_\mu \lambda w_1.y+z=w_1.
  \end{align*}
  The order of the original formula is \(2\) (since \(\Sum:
  (\INT\to(\INT\to\Prop)\to\Prop)\to\INT\to(\INT\to\Prop)\to\Prop\)),
  while the order of the formula obtained  by the translation is \(1\);
  note that \(\Sum_0: (\INT^2\to\Prop)\times (\INT^2\to\Prop)\times(\INT^2\to\Prop)
  \to \INT^2\to\Prop\).
  By further simplifications (note that the \(0\)-components \(\Sum_0,C_0,D_0,\ldots\) actually return
  \(\FALSE\) for the arguments arising from the main formula), we obtain:
  \begin{align*}
    & S_1 =_\mu \lambda w_1.\exists u_1.\Sum_1(\Plus_0,\Plus_1)\,n\,u_1\land C_1\,u_1\,w_1\\&
    C_1\,x =_\mu \lambda w_1.x<n\land 0=w_1\\&
    \Sum_1\,(f_0,f_1)\,x =_\mu \lambda z_1.x\le 0\land 0=z_1 \\&\qquad
    \lor
    x>0\land (f_0\,x\,z_1\lor \exists u_1.
    f_1\,x\,u_1\land D_1(f_0,f_1)\,x\,u_1\,z_1)\\&
    \Plus_0\,x =_\mu \lambda w_1.\FALSE\\&
    \Plus_1\,x =_\mu \lambda w_1.x+x=w_1\\&
    D_1\,(f_0,f_1)\,x\,y =_\mu \lambda w_1.
    \exists u_1.\Sum_1(f_0,f_1)\,(x-1)\,u_1
    \land E_1\,y\,u_1\,w_1\\&
    E_1\,y\,z =_\mu \lambda w_1.y+z=w_1.
  \end{align*}
  \qed
\end{example}

\fi

\begin{remark}
  The output, which uses tuples and simultaneous least fixpoints,
  can be converted to a formula in the original syntax of
\muhflz{} in polynomial time without increasing its order.
It suffices to flatten tuples into their
components, and encode the equations by one least fixpoint with an integer
tag selecting the equation and dummy arguments padding parameter lists
to a common type.
\qed
\end{remark}
\subsection{Correctness}
\label{sec:correctness}
We show the correctness of the translation.

The following lemma states that the output of the translation is well-typed.
\begin{lemma}
\label{lem:well-typedAndOccurrence}
If \(\stenv;\penv\pn\form:\fty\tr 
(\fm_\Mark,\seq{\fm}_{0,\dots,k+\gar(\fty)})\),
then \(\fromdisjTEp{\NT}, \fromdisjTE{\stenv}\pST 
(\fm_\Mark,\seq{\fm}_{0,\dots,k+\gar(\fty)}):
\fromdisjT{\fty}{k+2}\).
Also, for \((y\COL\fty) \in \stenv\),
\(y_{\Mark}\) does not occur free in \(\seq{\fm}_{0,\dots,k+\gar(\fty)}\).
\end{lemma}
\begin{proof}
  Straightforward induction on the derivation of
  \(\stenv;\penv\pn\form:\fty\tr 
(\fm_\Mark,\seq{\fm}_{0,\dots,k+\gar(\fty)})\).
\end{proof}

The following theorem states the correctness of the translation.
\begin{theorem}
  \label{th:fromdisj-correctness}
  If
 \( (D, S\,\lambda \seq{z}_{1,\ldots,\maxar}.\TRUE)\tr
             (D', \formalt)\), then 
 \(\sem{(D,S\,\lambda \seq{z}_{1,\ldots,\maxar}.\TRUE)}=\sem{(D',\formalt)}\).
\end{theorem}

The rest of this section is devoted to the proof of Theorem~\ref{th:fromdisj-correctness}.
The proof consists of the following two steps:
(i) we first reduce the proof of Theorem~\ref{th:fromdisj-correctness}
to the case where a given equation system
is recursion-free (in \cref{sec:reductionToRecFree}),
by using a standard technique of finite approximation,
and then (ii) we show the recursion-free case (in \cref{sec:CorrectnessRecursionFree}, with some preparation in \cref{sec:redd}).

For an equation system \((\NT,D, S\,\lambda \seq{z}.\TRUE)\),
we define \(\eqd\) as follows:
\(\fm \eqd \fa\) if
\(\semes{D}{\fm} = \semes{D}{\fa}\).
For \((F\,\seq{x} =_\mu \fm) \in D\),
we may drop the subscript \(\mu\)
and write \(F\,\seq{x} = \fm\)
if there is no confusion.
We write \([\fa_i/x_i]_{i=1}^{m}\fm\)
for the substitution
\([\fa_1/x_1,\dots,\fa_m/x_m]\fm\).

\subsubsection{Reduction to the Recursion-free Case}
\label{sec:reductionToRecFree}

Here we briefly explain how we can reduce
Theorem~\ref{th:fromdisj-correctness}
to the recursion-free case.
When taking finite approximations below, we assume that each equation
is written with enough formal parameters to make its right-hand side
have type \(\Prop\), using \(\eta\)-expansion if necessary.

For an equation system \((\NT,D,\fm_0)\) and \(m \in \Nat\),
the \emph{\(m\)-th approximation} \((\NT^{(m)},D^{(m)},\fm_0^{(m)})\)
is defined as follows:
\begin{align*}
\NT^{(m)} \defe&\
\set{F^{(i)} \mapsto \NT(F) \mid F \in \dom(\NT), 0 \le i \le m}
\\
\fm^{(i)}
\defe&\
[F^{(i)}/F]_{F\in\dom(\NT)}\fm
\quad(\text{for any \(\fm\) and \(i \in \set{0, \dots, m}\)})
\\
D^{(m)} \defe&\
\set{
F^{(i)}\,\seq{x}  =  \fm^{(i-1)}
 \mid (F\,\seq{x}  =  \fm) \in D, 1 \le i \le m}
\\\cup&\
\set{
F^{(0)}\,\seq{x}  = \FALSE \land \fm^{(0)}
 \mid (F\,\seq{x}  =  \fm) \in D}.
\end{align*}
For \(F^{(0)}\) above, we use \(\FALSE \land \fm^{(0)}\)
rather than \(\FALSE\),
in order to keep the form of \cref{eq:ESassumption}.
If \(\toform{(D, \fm_0)}\) is closed
and has type \(\Prop\), the argument in~\cite[Appendix~B.1]{ESOP2018full} yields
\[
\sem{(D,\fm_0)} = \LUB_{\Prop}\{\sem{(D^{(m)},\fm_0^{(m)})} \mid m \in \Nat\}.
\]

An equation system \((\NT,D,\fm_0)\)
is called \emph{recursion free} if
there is no cyclic dependency on \(D\).
More precisely, we define a binary relation \(\dep{}\) on \(\dom(\NT)\)
as follows:
\(F \dep{} F'\) iff \(F' \in \FVf(\fm)\)
where \((F\seq{x}=\fm) \in D\)
and \(\FVf(\fm)\) is defined by the following:
\begin{align*}
\FVf(x) &= \set{x},\\
\FVf( \form_1\lor \form_2 ) &=
\FVf( \form_1 ) \cup \FVf( \form_2 ),\\
\FVf( e_1 \le e_2\land\form) &=
\begin{cases}
\emptyset & (e_1 \le e_2 = \FALSE)
\\
\FVf( \form ) & (e_1 \le e_2 \neq \FALSE)
\end{cases}
,\\
\FVf( \form_1\, \form_2 ) &=
\FVf( \form_1 ) \cup \FVf( \form_2 ),\\
\FVf( \form\, e) &= \FVf( \form ).
\end{align*}
Then \(D\) is recursion free if the transitive closure \(\dep{}^*\) of \(\dep{}\) is irreflexive
(i.e., \(F \dep{}^* F\) for no \(F \in \dom(\NT)\)).
Clearly \((D^{(m)},\fm_0^{(m)})\) is recursion-free.

Now, since our translation is compositional,
we can easily show the following:
\begin{lemma}
 Suppose
\((D,S\,\lambda\seq{z}.\TRUE)
  \tr (D',\exists\seq{z}.S_1\,\seq{z})\).
If \((D^{(m)}, (S\,\lambda \seq{z}. \TRUE)^{(m)})
\tr (D_m, \fm_m)\),
then
\(\sem{(D_m, \fm_m)}
= \sem{(D'^{(m)}, (\exists \seq{z}. S_1\,\seq{z})^{(m)})}\).
\end{lemma}
\begin{proof}
Identify each symbol \((F^{(i)})_j\) in \(D_m\) with
\((F_j)^{(i)}\) in \(D'^{(m)}\).
Induction on the translation derivation shows that the translation
commutes with type-preserving renaming of fixpoint variables.
Hence, for \(i>0\), the corresponding defining equations coincide
up to renaming of bound variables.

For \(i=0\), rule \rn{Tr-Le} translates
\(\FALSE\land\form^{(0)}\) into components of the form
\(\lambda\seq{w}.(\FALSE\land\psi_j\,\seq{w})\).
These are constantly false predicates, so the corresponding
level-zero symbols denote the least elements of their types in
both systems.
Induction on \(i\) therefore shows that corresponding symbols
have the same denotation at every level \(0\le i\le m\).

Finally, by rule \rn{Tr-Main}, the main formula \(\fm_m\) is
\(\exists\seq{z}.(S^{(m)})_1\,\seq{z}\), which corresponds to
\((\exists\seq{z}.S_1\,\seq{z})^{(m)}\) under the above identification.
Thus, the claimed equality follows.
\end{proof}

We can reduce the proof of Theorem~\ref{th:fromdisj-correctness}
to the recursion-free case as follows.
Let
\((D, S\,\lambda \seq{z}. \TRUE)
\tr (D', \exists \seq{z}. S_1\,\seq{z})\)
and
\((D^{(m)}, (S\,\lambda \seq{z}. \TRUE)^{(m)})
\tr (D_m, \fm_m)\); then
\begin{align*}
\sem{(D,S\,\lambda \seq{z}. \TRUE)} 
&= \LUB_{\Prop}\{\sem{(D^{(m)},(S\,\lambda \seq{z}. \TRUE)^{(m)})} \mid m \in \Nat\}
\\
&= \LUB_{\Prop}\{\sem{(D_{m},\fm_{m})} \mid m \in \Nat\}
\\
&= \LUB_{\Prop}\{\sem{(D'^{(m)}, (\exists \seq{z}. S_1\,\seq{z})^{(m)})} \mid m \in \Nat\}
\\
&= \sem{(D', \exists \seq{z}. S_1\,\seq{z})}
\end{align*}
where for the second equation we assume the recursion-free case.

\ifdraft
\newpage
\fi

\subsubsection{Reduction Relation with Explicit Substitution}
\label{sec:redd}

In our proof of the recursion-free case,
we show a subject reduction property.
To this end,
we modify the reduction strategy
by using explicit substitution, keeping the adequacy for the semantics.
For this modification, we first extend the syntax of formulas as follows:
\begin{equation}
\label{eq:ExtendWithExpSubst}
\begin{aligned}
    \form  ::=
\ &x \mid \form_1\lor \form_2 \mid e_1\le e_2\land\form
\mid \form_1\form_2 \mid \form\, e
\\
\mid\, &\ess{x_1}{\fm_1}{x_m}{\fm_m}\fm
\end{aligned}
\end{equation}
Here \(\ess{x_1}{\fm_1}{x_m}{\fm_m}\fm\) is called
an \emph{explicit substitution},
and limited to ground types as follows:
\infrule[T-ESub]{
\stenv\pST \fm_i:\INT^{\maxar}\to\Prop
\quad(i=1,\dots,m)
\\
\stenv,
x_1\COL\INT^{\maxar}\to\Prop,\dots,x_m\COL\INT^{\maxar}\to\Prop
\pST \fm:\Prop
}{
\stenv\pST \ess{x_1}{\fm_1}{x_m}{\fm_m}\fm:\Prop
}
Its meaning is given through
\[
\toform{(D,\ess{x_1}{\fm_1}{x_m}{\fm_m}\fm)}
\defe \toform{(D,\oss{x_1}{\fm_1}{x_m}{\fm_m}\fm)}.
\]
Thus explicit substitution has the same meaning as 
ordinary substitution,
but delays substitution
until we need \(\fm_i\)
for further reduction.
As in the definition of \(\redd\) below,
while we use ordinary substitutions 
for \(\beta\)-redex to which we can apply
\rnp{Tr-App} and \rnp{Tr-AppI},
we use explicit substitution for
those corresponding to \rnp{Tr-AppG}
because 
the argument after the translation by \rnp{Tr-AppG}
is never substituted.

We extend the translation by adding the following rule
for explicit substitutions:
\infrule[Tr-ESub]{\mspace{10mu}
\begin{aligned} &
\stenv;\,\penv
\pn 
\fp_{i}:{\INT}^{\maxar}{\to}\Prop\tr (\fp_{i,\Mark},\fp_{i,0},\ldots,\fp_{i,k})
\mspace{-100mu}\\
&&(i=1,\dots,m)&
\\&
\stenv;\,\penv,\,\seq{x'}_{1,\dots,m} \pn 
\fm:\Prop\tr
(\fm_\Mark,\fm_0,\ldots,\fm_{k+m})
\mspace{-100mu}\\&
\fa_j
=
\lambda \seq{w}_{1,\ldots,\maxar}.\,
\fm_j\,\seq{w}
\lor \textstyle\bigvee_{i=1}^{m}
 \exists \seq{u}_{1,\ldots,\maxar}.
\big(
 \fm_{k+i}\,\seq{u}
 \land \fp_{i,j}\,\seq{u}\,\seq{w}
\big)
\mspace{-100mu}\\
&&(j=\Mark,0,\dots,k)&
\end{aligned}
\mspace{-10mu}
}{
\stenv;\,\penv
 \pn 
\ess{x'_1}{\fp_1}{x'_m}{\fp_m}
\fm
:\Prop\tr
(\fa_\Mark,\fa_0,\ldots,\fa_{k})
}
In the rest of this section,
by a formula we mean a formula that may contain explicit substitutions,
except for formulas in an equation system
and except for the case where we explain explicitly.
Note that output formulas of the extended translation never contain explicit substitutions.

Let \((\NT,D, S\,\lambda \seq{z}.\TRUE)\) be an equation system.
For decomposing actual arguments \(\fea_1,\dots,\fea_{m}\) of
a function \(F \in \dom(\NT)\)---recall that \(\fea_i\) ranges over formulas and integer expressions---we 
define \(\decompA(\fea_1,\dots,\fea_{m'}, \NT(F))\) 
in the same way as \(\decomparg\) as follows:
\begin{align*}
  & \decompA(\epsilon, \Prop) = (\epsilon, \epsilon, \epsilon)\\
  & \decompA(\fea\cdot \seq{\feb}, \sty\to\fty) =\\
&  
  \left\{
  \begin{array}{ll}
    (\fea\cdot \seq{\fm}, \seq{\fa}, \seq{e}) &\mbox{if 
      \(\decompA(\seq{\feb},\fty)=(\seq{\fm}, \seq{\fa}, \seq{e}), \seq{\fm}\ne \epsilon\)}\\
    (\fea, \seq{\fa}, \seq{e}) &\mbox{if \(\ord(\sty)>0,
      \decompA(\seq{\feb},\fty)=(\epsilon, \seq{\fa},\seq{e})\)}\\
    (\epsilon, \fea\cdot \seq{\fa}, \seq{e}) &\mbox{if \(\sty=\INT^\maxar\to\Prop\),}\\
    & \  \mbox{\(\decompA(\seq{\feb},\fty)=(\epsilon, \seq{\fa},\seq{e})\)}\\
    (\epsilon, \seq{\fa}, \fea\cdot\seq{e}) &\mbox{if \(\sty=\INT,
      \decompA(\seq{\feb},\fty)=(\epsilon, \seq{\fa},\seq{e})\)}
  \end{array}
  \right.
\end{align*}

Now we define the modified reduction relation \(\redd\)
for a formula \(\fm\) 
such that
\(\NT,x_1 \COL \INT^{\maxar}{\to}\Prop,\ldots,x_k \COL \INT^{\maxar}{\to}\Prop \pST \fm : \Prop\) holds
for some \(x_1,\ldots,x_k\).
We define the set of \emph{evaluation contexts} by:
\[
E ::= \hole \mid E \lor \fm \mid \fm \lor E \mid
\ess{x_1}{\fm_1}{x_m}{\fm_m}E.
\]
Then \(\redd\) is defined by the following rules:
\infrule{\sem{\pST e_1:\INT}  >  \sem{\pST e_2:\INT}
\andalso (e_1\le e_2) \neq \FALSE
}{E[e_1\le e_2 \land \fm] \redd E[\FALSE\land \fm]}
\infrule{\sem{\pST e_1:\INT} \le \sem{\pST e_2:\INT}
}{E[e_1\le e_2 \land \fm] \redd E[\fm]}
\infrule{(F\,w_1\,\cdots\,w_m = \fm) \in D
\\ \decompA(\fea_1,\cdots,\fea_m,\NT(F))=
(\seq{\fm},\seq{\fa},\seq{e})
\\ \decomparg(w_1,\cdots,w_m,\NT(F))=
(\seq{y}\COL\seq{\sty},\seq{x},\seq{z})
\\ \seq{x} \text{ do not occur in }E[F\,\fea_1\,\cdots\,\fea_m]
}{E[F\,\fea_1\,\cdots\,\fea_m] \redd 
  E[\esubst{\seq{x}}{\seq{\fa}}[\seq{e}/\seq{z}][\seq{\fm}/\seq{y}]\fm]}
\infrule{}{E[\esubst{\seq{x}}{\seq{\fm}}(x_i\,\seq{e})] \redd E[\fm_i\,\seq{e}]}
\infrule{x \notin \{x_1,\dots,x_{|\seq{x}|}\}\cup\dom(\NT)}{E[\esubst{\seq{x}}{\seq{\fm}}(x\,\seq{e})] \redd E[x\,\seq{e}]}
\infrule{}{E[\esubst{\seq{x}}{\seq{\fm}}(\fa_1\lor\fa_2)]
 \redd E[ (\esubst{\seq{x}}{\seq{\fm}}\fa_1)
     \lor (\esubst{\seq{x}}{\seq{\fm}}\fa_2) ]}
\infrule{}{E[\esubst{\seq{x}}{\seq{\fm}}(\FALSE \land \fm)]
 \redd E[ \FALSE \land (\esubst{\seq{x}}{\seq{\fm}} \fm) ]}
Note that the above reduction preserves the form of~\eqref{eq:ExtendWithExpSubst}
and hence the applicability of the translation \(\tr\).
For any \(\fm\), \(\fm\) is a normal form with respect to \(\redd\)
iff \(\fm\) is generated by:\asd{todo: prove again}
\begin{equation}
\label{eq:normalFormRedd}
\nf ::= x \, \seq{e} \ (x \notin \dom(\NT)) \mid \FALSE \land \fm \mid \nf \lor \nf.
\end{equation}
Clearly we have:
\begin{lemma}
\label{lem:sound_redd}
If \(\fm \redd \fa\),
then 
\(\semes{D}{\fm} = \semes{D}{\fa}\).
\end{lemma}
%

\subsubsection{Correctness in the Recursion-free Case}
\label{sec:CorrectnessRecursionFree}

It remains to show the correctness in the recursion-free case,
which follows from the subject reduction property below.
For a type \(\fty = \sty_1 \to \cdots \to \sty_n \to \fty'\),
we write \(\btype{\fty}{n}\) for \(\fty'\).

\begin{lemma}[subject reduction]
  \label{lem:subjectReduction}
Suppose that we have
\((D, S\,\lambda \seq{z}.\TRUE)\tr
 (D', \exists \seq{z}.S_1\,\seq{z})\).
If \(\fm \redd \fa\) and
\[
\seq{x}_{1,\ldots,k}
\pn \fm:\Prop\tr (\fm_\Mark,\fm_0,\ldots,\fm_k),
\]
then there exist \(\fa_\Mark, \fa_0,\ldots,\fa_{k}\) such that
\[
\seq{x}_{1,\ldots,k}
\pn \fa:\Prop\tr (\fa_\Mark,\fa_0,\ldots,\fa_k)
\]
and 
\(\fm_i \eqdp \fa_i\)
for each \(i\in\set{\Mark,0,\ldots,k}\).
\end{lemma}
As the proof of the above lemma is quite technical and long,
we defer it to Appendix~\ref{sec:proof-subj}.

The following lemma states the correctness in the recursion-free case,
from which Theorem~\ref{th:fromdisj-correctness} follows.
\begin{lemma}
\label{lem:fromdisj-correctRecFree}
Suppose that 
\((D, S\,\lambda \seq{z}_{1,\ldots,\maxar}.\TRUE)\)
is a recursion-free equation system.
If
 \( (D, S\,\lambda \seq{z}_{1,\ldots,\maxar}.\TRUE)\tr
             (D', \exists \seq{z}.\,S_1\,\seq{z})\), then 
 \(\sem{(D,S\,\lambda \seq{z}_{1,\ldots,\maxar}.\TRUE)}=\sem{(D',\exists \seq{z}.\,S_1\,\seq{z})}\).
\end{lemma}
\begin{proof}
Let the rule of \(S\) be \(S\,x = \fm\);
then \(D'\) has the rule \(S_1 = \fm_1\) where \(S_1\) and \(\fm_1\) have type \(\INT^{\maxar}\to\Prop\).
Since 
\(\semes{D}{ S\,\lambda \seq{z}_{1,\ldots,\maxar}.\TRUE} = 
\semes{D}{\subst{x}{\lambda \seq{z}_{1,\ldots,\maxar}.\TRUE}\fm}\),
it suffices to show
\[
\semes{D}{\subst{x}{\lambda \seq{z}_{1,\ldots,\maxar}.\TRUE}\fm} = \semes{D'}{\exists \seq{z}.\,\fm_1\,\seq{z}}.
\]

Since \((D, S\,\lambda \seq{z}_{1,\ldots,\maxar}.\TRUE)\) is recursion-free,
\(\fm\) has a normal form \(\nf\) with respect to \(\redd\).
We have \(\fm \redds \nf\)
and let 
\(x \pn \fm:\Prop\tr (\fm_\Mark,\fm_0,\fm_1)\);
then by subject reduction (Lemma~\ref{lem:subjectReduction}),
we have 
\(x \pn \nf:\Prop
\allowbreak\tr (\nf_\Mark,\nf_0,\nf_1)\)
and 
\(\semes{D'}{\fm_i} = \semes{D'}{\nf_i}\) \((i=\Mark,0,1)\).
Also, we have \(\semes{D}{\fm} = \semes{D}{\nf}\)
by Lemma~\ref{lem:sound_redd},
and hence \(\semes{D}{\subst{x}{\lambda \seq{z}_{1,\ldots,\maxar}.\TRUE}\fm} = \semes{D}{\subst{x}{\lambda \seq{z}_{1,\ldots,\maxar}.\TRUE}\nf}\).
Therefore we can assume that \(\fm\) is a normal form
without loss of generality.

Then we can directly check
\(\semes{D}{\subst{x}{\lambda \seq{z}_{1,\ldots,\maxar}.\TRUE}\nf} = 
\semes{D'}{\exists \seq{z}.\,\nf_1\,\seq{z}}\)
by induction on
the structure of normal forms \eqref{eq:normalFormRedd}
in \cref{sec:redd}.
%
\end{proof}

\section{Applications}
\label{sec:app}

As mentioned already, the translation from order-\(n\) reachability games to
order-(\(n+1\)) may-reachability enables us to use automated (un)reachability checkers
for solving the reachability game problem, and
the translation in the other direction enables us to use, for example, reachability game solvers for
non-higher-order programs as a may-reachability checker for order-1 programs.

As a direct application of the former translation, we have applied it to
the \nuhflz{} solver \rethfl{}~\cite{DBLP:conf/aplas/KatsuraIKT20},
which is a refinement-type-based validity checker for formulas of
\nuhflz{}, the fragment of \hflz{} without least fixpoint operators (but with greatest fixpoint operators).
The fragment \nuhflz{} is dual to \muhflz{}, in the sense that,
for every closed formula \(\form\) of type \(\Prop\) of \muhflz{}, there exists a \nuhflz{} formula
\(\dual{\form}\) such that \(\form\) is valid if and only if \(\dual{\form}\) is invalid, and vice versa;
\(\dual{\form}\) is obtained from \(\form\) by just replacing each logical operator (including
fixpoint operators) with its de Morgan
dual, and \(e_1\le e_2\) with \(e_1>e_2\).
Using a refinement type system, \rethfl{} reduces the validity of a given \nuhflz{} formula
in a sound (but incomplete) manner
to an extended CHC (constrained Horn clauses) problem, where disjunction is allowed in
the head of each clause, and passes the problem to an extended CHC solver called \pcsat{}~\cite{DBLP:conf/aaai/Satake0Y20}. For a fragment of \nuhflz{} corresponding to \emph{disjunctive} \muhflz{},
however, the reduced problem is actually an ordinary CHC problem, for which
more efficient tools~\cite{DBLP:journals/fmsd/KomuravelliGC16,Eldarica,DBLP:journals/jar/ChampionCKS20}
can be invoked. Thus, we can use the translation in Section~\ref{sec:todisj} to improve the
efficiency of \rethfl{}.

From the benchmark suite of \rethfl{}~\cite{DBLP:conf/aplas/KatsuraIKT20}
\ifshort
\else
(which originates from 
\cite{DBLP:conf/sas/IwayamaKST20}, \url{https://github.com/Hogeyama/hfl-benchmark/tree/master/inputs/hfl/HO-nontermination}),
\fi
we picked the ``non-termination'' benchmark set, which consists of formulas
obtained from non-termination verification of higher-order programs. All the formulas
in that benchmark set do not belong to (the dual of) disjunctive \muhflz{}
(in contrast, the problems in the other benchmark sets belong to disjunctive \muhflz{},
hence our translation is not required).
We have implemented the translation in Section~\ref{sec:todisj}, applied it
to the problems in the ``non-termination'' benchmark set, and then ran
\rethfl{} with a CHC solver \hoice{}~\cite{DBLP:journals/jar/ChampionCKS20,DBLP:conf/aplas/Champion0S18} as the back-end solver. We have compared the result
with plain \rethfl{} (without the transformation), which uses the extended CHC solver \pcsat{}.

\newcommand\TO{\textrm{timeout}}
\newcommand\UN{\textrm{unknown}}
\begin{table}
  \caption{Experimental results. Times are 
    in seconds, with the timeout of 180 seconds.}
\label{tab:rethfl}
\begin{center}
\begin{tabular}{|l|r|r|r|}
  \hline
  input & \rethfl{} & {\small \rethfl{}+i.s.} & {\small \rethfl{}+ tr.} \\
  \hline
  \hline
fixpoint\_nonterm & 11.579 & 0.054 &  0.102 \\\hline
unfoldr\_nonterm & \TO{} & \UN{} & 4.22  \\\hline
indirect\_e & 16.832 & 0.035 & 0.066  \\\hline
alternate & \UN{}& \UN{}& \UN{} \\\hline
fib\_CPS\_nonterm & \TO{} & 0.047 & 0.075 \\\hline
foldr\_nonterm & 8.447 & \UN{} & 0.122\\\hline
passing\_cond & 116.423 & \UN{} & 0.444\\\hline
indirectHO\_e & 11.582 & 0.044 & 0.073 \\\hline
inf\_closure & \TO{} & 20.171 & 9.080 \\\hline
loopHO & \TO{} & 0.026 & 0.121\\
 \hline
  \end{tabular}
\end{center}
\end{table}

The results are summarized in Table~\ref{tab:rethfl}.
The column '\rethfl{}' shows the result of plain \rethfl{} with
\pcsat{} as the back-end extended CHC solver (since ordinary CHC solvers
are inapplicable to this benchmark set, as explained above).
The column '\rethfl{}+i.s.'
shows the result of \rethfl{} where the subtyping relation
has been replaced by the imprecise one (equivalent to that of \horus{}~\cite{DBLP:journals/pacmpl/BurnOR18},
a HoCHC solver
that can also be viewed as a \nuhflz{} solver)
so that the type checking problem is reduced to ordinary CHC solving.
The column '\rethfl{}+tr.' shows the result of
\rethfl{} with our translation. In both
'\rethfl{}+i.s.' and '\rethfl{}+tr.', \hoice{} was used as the back-end CHC solver.
The entry ``\UN{}'' indicates that the solver terminated with
the answer ``ill-typed'', in which case, we do not know whether
the formula is valid or invalid, due to the incompleteness of
the underlying refinement type system.\footnote{Although the understanding of the refinement type system of
\rethfl{}
is not required below, interested readers may wish to consult 
\cite{DBLP:conf/aplas/KatsuraIKT20}.}
The refinement type system used in
'\rethfl{}+i.s.' is less precise than the one used in \rethfl{}; hence,
it returns more \UN{}s.
As clear from the table, our translation significantly improved the efficiency
of \rethfl{}.

The translation in the other direction given in Section~\ref{sec:fromdisj}
also 
helps \rethfl{}, especially for relaxing the limitation
caused by the incompleteness of the underlying refinement type system.
For example, consider the formula \(S\,\TRUE\), where:
\newcommand\App{\mathit{App}}
\begin{align*}
  S\,t &=_\mu \App\,(\lambda x.x\ne 0\land t)\,0\qquad
  \App\,p\,y =_\mu p\,y \lor \App\,(\lambda z.p(z-1))\,(y+1).
\end{align*}
The formula is invalid, but \rethfl{} (nor \horus{}~\cite{DBLP:journals/pacmpl/BurnOR18},
a higher-order CHC solver based on a refinement type system)
  cannot prove the validity
  of the dual formula, due to the incompleteness of the refinement type
\ifshort
\else
system
(which is related to
the incompleteness of a refinement type system addressed by \cite{Unno18POPL}
by inserting extra arguments). \fi
The translation in Section~\ref{sec:fromdisj}
  yields 
the following order-0 formula:\footnote{We have implemented a
  prototype translator, but have not yet integrated it into \rethfl{}.
  For readability,
  here we show the formula obtained by
  some manual simplification of the automatically generated formula.}
\begin{align*}
&  S_1 =_\mu \exists x.\App_1\,0\,x \land x\ne 0\\
& \App_1\,y\,z =_\mu y=z \lor \exists w.\App_1\,(y+1)\,w\land w-1=z.
\end{align*}
Here, \(\App_1\,y\,z\) intuitively means that \(\App\,p\,y\) can be reduced to
\(p\,z\).
The underlying type system of \rethfl{} is complete for order-0 formulas,
and indeed, the order-0 formula above
can automatically be proved invalid by \rethfl{}.

\section{Related Work}
\label{sec:related}

The relationship between order-\(n\) reachability games and order-(\(n+1\)) may-reachability
has some deep connection to the relationship between order-\(n\) tree languages and
order-(\(n+1)\) word languages~\cite{DammTCS,Asada16,DBLP:conf/fscd/Asada020},
intuitively because the may-reachability problem is concerned about the set of ``paths'' of the execution tree of
a given program, whereas the reachability game problem
is also concerned about the branching structures of the execution tree.
Indeed, our translations (especially, the use of \(\form_\Mark\) and \(\form_0\) components
in the translation in Section~\ref{sec:fromdisj}) have been inspired by Asada and Kobayashi's
translations between tree and word languages~\cite{DBLP:conf/fscd/Asada020}. Kobayashi et al.~\cite{KDG19LICS}
have also used a similar idea for a
characterization of termination probabilities of higher-order probabilistic programs.

For finite-data programs (programs in Section~\ref{sec:reachability-as-hfl} without integers),
according to the complexity results on HORS model checking~\cite{Ong06LICS,KO11LMCS},
both the order-\(n\) reachability game problem and the order-(\(n+1\)) may-reachability game problem
are \(n\)-EXPTIME complete, which implies that there are mutual translations between them. Concrete
translations have, however, not been given (except unnatural translations through Turing machines).
Also, the complexity-theoretic argument for the existence of translations does not apply in the
presence of integers.

For HORS model checking, Parys~\cite{DBLP:conf/icalp/Parys21} developed an
order-decreasing transformation for higher-order grammars,
which shares some ideas with our translation in Section~\ref{sec:fromdisj}.
The details of the translations are however quite different. His translation makes use of
finiteness in a crucial manner, and is not applicable in the presence of integers. Also,
his translation is not size-preserving.

For order-1 programs, Kobayashi et al.~\cite{DBLP:conf/sas/0001NIU19} have shown that
linear-time omega regular properties can be translated to order-0 \hflz{} formulas.
Our translation in Section~\ref{sec:fromdisj} may be viewed as a higher-order extension
of their translation, while the properties are restricted to may-reachability.

The fragment \muhflz{} (or its dual fragment \nuhflz{}) is essentially (modulo the restriction
of data domains to integers) equivalent to HoCHC~\cite{DBLP:journals/pacmpl/BurnOR18},
a higher-order extension of CHC. Therefore, the result of this paper should be useful also
for improving HoCHC solvers.

\section{Conclusion}
\label{sec:conc}

We have shown translations between order-\(n\) reachability games and
order-(\(n+1\)) may-reachability, and proved their correctness. We have applied
the translations to higher-order program verification, and obtained promising
results in preliminary experiments.
As mentioned in Section~\ref{sec:related}, our results are closely related
to the correspondence between higher-order word and tree languages~\cite{DBLP:conf/fscd/Asada020}.
A deeper investigation of the relationship and generalization of the translations
that subsume the related translations~\cite{DBLP:conf/fscd/Asada020,KDG19LICS} are left for future work.

\subsection*{Acknowledgments}
  We would like to thank anonymous referees for useful comments.
This work was supported by
JSPS KAKENHI Grant Number JP20H05703, JP24K14814 and
JST SPRING Grant Number JPMJSP2108.

\ifshort\else
\newpage
\appendix
\section*{Appendix}
\section{Proof of Lemma~\ref{lem:subjectReduction}}
\label{sec:proof-subj}
We first prepare two substitution lemmas
that correspond to
the application rules \rnp{Tr-App} and \rnp{Tr-AppI}.

The following is the substitution lemma corresponding to \rnp{Tr-AppI}:
\begin{lemma}[Substitution lemma (integer)]
\label{lem:substitutionInt}
Let \(\fm\) be a formula that does not contain an explicit substitution,
\(e\) be a closed integer expression, and
\[
 \stenv,\,
 z\COL\INT;\,
 \penv
 \pn \fm:\fty \tr (\fm_\Mark,\fm_0,\ldots,\fm_{k+m}).
\]
where 
\(m=\gar(\fty)\).
Then we have
\[
 \stenv;
 \penv
 \pn [e/z]\fm:\fty \tr 
([e/z]\fm_\Mark,[e/z]\fm_0,\ldots,[e/z]\fm_{k+m}).
\]
\end{lemma}
\begin{proof}
By straightforward induction on \(\fm\).
\end{proof}

\ifdraft\clearpage\fi

Next we show the substitution lemma corresponding to \rnp{Tr-App}.
First we prepare some definitions and a lemma.

For a formula \(\fm\), we write
\(\caa{\fm}{k}{m}\) for \((\fm_0,\fm_{k+1},\dots,\fm_{k+m})\).
Note that the translation result
of \(\fm\,\fa\) in \rnp{Tr-App} in \cref{fig:tr-n-1}
can be written as the following:
\begin{alignat*}{10}
\big(\fm_\Mark(\fa_\Mark,\caa{\fa}{k}{m'}),\,
\fm_0(\fa_0,\caa{\fa}{k}{m'}),\,
&
\fm_1(\fa_1,\caa{\fa}{k}{m'}),
&&\ldots,\,&&
\fm_k(\fa_k,\caa{\fa}{k}{m'})&&,
\\&
\fm_{k+1}(\caa{\fa}{k}{m'}),
&&\ldots,\,&&
\fm_{k+m}(\caa{\fa}{k}{m'})&&\big)
\end{alignat*}

The following can be shown easily by induction on \(\fm\).
\begin{lemma}[weakening]
\label{lem:weakening}
If 
\begin{align*} &
\stenv;
\penv
\pn 
\fm: \fty \tr 
(\form_\Mark,\form_0,
\form_{1},\ldots,\form_{k},
\form_{k+1},\ldots,\form_{k+m}
)
\end{align*}
then
\begin{align*} &
\stenv;\penv,x
\pn 
\fm: \fty \tr 
(\form_\Mark,\form_0,
\form_{1},\ldots,\form_{k},
\form_0,
\form_{k+1},\ldots,\form_{k+m}
).
\end{align*}
\end{lemma}

Now we show the substitution lemma.
Here we consider simultaneous substitution,
because we cannot apply this lemma repeatedly since
\([\seq{\fa} / \seq{y}] \fm\) below may contain an explicit substitution.
\begin{lemma}[Substitution lemma (higher-order)]
\label{lem:substitutionHigher}
Let \(\fm\) be a formula that does not contain an explicit substitution,
and
\begin{align*}
&
 \seq{y}_{1,\dots,q}\COL\seq{\sty};\
\seq{x'}_{1,\dots,m}
 \pn 
\fm:\fty\tr (\fm_\Mark,\fm_0,\ldots,\fm_{m+\gar(\fty)})
\\&
m'_i=\gar(\sty_i)
\quad
\decomp(\sty_i) = (\seq{\sty_i},m'_i,p_i)
\quad(i=1,\dots,q)
\\&
\penv \pn \fa_i : \sty_i \tr
(\fa_{i,\Mark},\fa_{i,0},\dots,\fa_{i, k + m'_i})
\quad(i=1,\dots,q)
\\&
\seq{y}_i = (y_{i,\Mark},y_{i,0},\dots,y_{i,\gar(\sty_i)})
\qquad
\seq{y}^\circ_i = (y_{i,0},\dots,y_{i,\gar(\sty_i)})
\\&
\begin{alignedat}{3}
\theta^{j}
&=
\big[& (\fa_{i,j},\caa{\fa_{i}}{k}{m'_{i}}) / \seq{y}_{i} &\big]_{i=1}^{q}
\qquad
&&(j=\Mark,0,\dots,k)
\\
\theta^\circ
&=
\big[& \caa{\fa_{i}}{k}{m'_{i}} / \seq{y}^\circ_{i} &\big]_{i=1}^{q}
\end{alignedat}
\\&
\penv, \seq{x'}_{1,\dots,m} \pn 
[\seq{\fa} / \seq{y}] \fm
:\fty\tr
(\fm^{\circ}_{\Mark},\fm^{\circ}_{0},\dots,\fm^{\circ}_{k+m+\gar(\fty)}).
\end{align*}
Then we have:
\begin{enumerate}
\item
\label{item:subOfSubstLemma}
\(
\theta^0 \fm_\Mark
=
\theta^0 \fm_0
\).
\asdm{This equation is the reason why we treat nonterminals separately in the definition of the
     transformations rules.
Without the separate treatment, we need to replace this equality with
logical relation of \(\approx\), which seems to work(?).}\asdm{If \(\subst{\seq{y}}{\seq{\fm}}\fp\) contains a non-ground free variable \(y\) other than
     nonterminals, this equation does not hold (even with the separate treatment for
     nonterminals and even with relaxing \(=\) to \(\approx\)).}
\item
\label{item:mainOfSubstLemma}
\(
\mspace{-20mu}
\begin{alignedat}[t]{20}
(&\fm^{\circ}_{\Mark}&&,
\fm^{\circ}_{0}&&,
\fm^{\circ}_{1}&&,\dots&&,\fm^{\circ}_{k}&&,
\fm^{\circ}_{k+1}&&,
\dots&&,
\fm^{\circ}_{k+m+\gar(\fty)}
&&)
\\
\beeq
(&\theta^{\Mark}\fm_{\Mark}&&,
\theta^{0}\fm_{\Mark}&&,
\theta^{1}\fm_{\Mark}&&,
\ldots&&,
\theta^{k}\fm_{\Mark}&&,
\theta^{\circ}\fm_{1}&&,
\ldots&&,
\theta^{\circ}\fm_{m+\gar(\fty)}
&&).
\end{alignedat}
\)
\asdm{Here \(\beeq\) is used just for \(\beta\eta\)
of auxiliary \(\lambda\)-abstraction,
and in fact tex-macro is \(\backslash{}\)beeq rather than \(\backslash{}\)eqdp,
while in the proof of subject reduction we use \(\beeq\)
for reduction of nonterminal etc.}%
\asdm{In this proof we use the other item
in the case of \rnp{Tr-App} (not \rnp{Tr-AppG}).}
\end{enumerate}
\end{lemma}
\begin{proof}
We can show \cref{item:subOfSubstLemma} easily
by induction on \(\fm\) and case analysis
on the last rule used for the derivation
\(
 \seq{y}_{1,\dots,q}\COL\seq{\sty};\
\seq{x'}_{1,\dots,m}
 \pn 
\fm:\fty\tr (\fm_\Mark,\fm_0,\ldots,\fm_{m+\gar(\fty)})
\).
We show \cref{item:mainOfSubstLemma} by the same induction and case analysis.
Basically the proof is straightforward,
where we use the latter part of Lemma~\ref{lem:well-typedAndOccurrence}.
Here we show only the cases of \rnp{Tr-Var} and \rnp{Tr-App};
in the latter case, we use \cref{item:subOfSubstLemma}.

\myitem
Case of \rnp{Tr-Var}:
Let the last rule be the following,
where \(i \in \set{1,\dots,q}\):
\[
\infer{
\begin{aligned}[t] &
\seq{y}\COL\seq{\sty}; 
\seq{x'}_{1,\dots,m}
 \pn 
(\fm =)\ y_i: \sty_i 
\\&\quad
\tr\big(
\begin{alignedat}[t]{10}
&(\fm_\Mark&&,\fm_0&&,
\fm_{1}&&,\ldots&&,\fm_{m}&&,
\fm_{m+1}&&,\ldots&&,\fm_{m+\gar(\fty)}&&) =\big)
\\
&(y_{i,\Mark}&&, y_{i,0}&&,
 y_{i,0}&&,\ldots&&,y_{i,0}&&,
 y_{i,1}&&,\ldots&&,y_{i,m'_i}&&)
\end{alignedat}
\end{aligned}
}{
\begin{aligned}[t] &
\decomp(\sty_i) = (\seq{\sty_i}, m'_i ,p_i)
\end{aligned}
}
\]

By the weakening lemma (Lemma~\ref{lem:weakening}),
we have 
\begin{align*} &
\penv, \seq{x'}_{1,\dots,m}
 \pn 
([\seq{\fa} / \seq{y}] \fm =)\ \fa_i : \sty_i \tr
\\&
\begin{alignedat}{20}
\big(
&(\fm^{\circ}_{\Mark}&&,\fm^{\circ}_{0}&&,
\fm^{\circ}_{1}&&,\dots&&,\fm^{\circ}_{k}&&,
\fm^{\circ}_{k+1}&&,\dots&&,\fm^{\circ}_{k+m}&&,
\fm^{\circ}_{k+m+1}&&,\dots&&,\fm^{\circ}_{k+m+m'_i}&&) =\big)
\\
&(\fa_{i,\Mark}&&,\fa_{i,0}&&,
\fa_{i,1}&&,\dots&&,\fa_{i,k}&&,
\fa_{i,0}&&,\dots&&,\fa_{i,0}&&,
\fa_{i,k+1}&&,\dots&&,\fa_{i, k + m'_i}&&)
\end{alignedat}
\end{align*}
Then we can check the required equation component-wise.

\myitem
Case of \rnp{Tr-App}:
Let the last rule be the following:
\[
\infer{\begin{aligned}[t] &
\seq{y}\COL\seq{\sty}; \seq{x'}_{1,\dots,m}
            \pn (\fm=)\ \fm'\,\fa':\fty \tr
\\&\quad
\big((\fm_\Mark,\fm_0,
\fm_{1}\ldots,\fm_{m},
\fm_{m+1}\ldots,\fm_{m+\gar(\fty)}) =\big)
\\&\quad\phantom{\big(}
\big(\fm'_\Mark(\fa'_\Mark,\caa{\fa'}{m}{m''}),\,
\fm'_0(\fa'_0,\caa{\fa'}{m}{m''}),\,
\\&\quad
\begin{alignedat}{10}
\phantom{\big(\big(}&\fm'_1(\fa'_1,\caa{\fa'}{m}{m''}),
&&\ldots,\,&&
\fm'_m(\fa'_m,\caa{\fa'}{m}{m''})&&,
\\\phantom{\big(\big(}&
\fm'_{m+1}(\caa{\fa'}{m}{m''}),
&&\ldots,\,&&
\fm'_{m+m'}(\caa{\fa'}{m}{m''})&&\big)
\end{alignedat}
\end{aligned}
}{
\begin{aligned}[t] &
 \ord(\sty'\to\fty)>1\andalso   \gar(\sty'\to\fty)=m'
    \andalso
\gar(\sty') = m''
\\&
    \seq{y}\COL\seq{\sty};
    \seq{x'}_{1,\dots,m}
    \pn
 \fm':\sty'\to \fty
 \tr (\fm'_\Mark,\fm'_0,\ldots,\fm'_{m+m'})
\\&
    \seq{y}\COL\seq{\sty};
    \seq{x'}_{1,\dots,m} \pn
 \fa':\sty'\tr
    (\fa'_\Mark,\fa'_0,\ldots,\fa'_{m+m''})
\end{aligned}}
\]
Here note that we have \(\gar(\fty) = \gar(\sty'\to\fty) = m'\)
since \(\ord(\sty'\to\fty)>1\).

By induction hypothesis,
there exist 
\(\fm'^{\circ}_{\Mark},\fm'^{\circ}_{0},\dots,\fm'^{\circ}_{k+m+m'}\)
such that
\begin{align*} &
\penv, \seq{x'}_{1,\dots,m}
 \pn 
[\seq{\fa} / \seq{y}]
\fm'
:\sty'\to\fty\tr
\\&
\begin{alignedat}{20}
(&\fm'^{\circ}_{\Mark}&&,
\fm'^{\circ}_{0}&&,
\fm'^{\circ}_{1}&&,\dots&&,\fm'^{\circ}_{k}&&,
\fm'^{\circ}_{k+1}&&,
\dots&&,
\fm'^{\circ}_{k+m+m'}
&&)
\\
\beeq
(&\theta^{\Mark}\fm'_{\Mark}&&,
\theta^{0}\fm'_{\Mark}&&,
\theta^{1}\fm'_{\Mark}&&,
\ldots&&,
\theta^{k}\fm'_{\Mark}&&,
\theta^{\circ}\fm'_{1}&&,
\ldots&&,
\theta^{\circ}\fm'_{m+m'}
&&),
\end{alignedat}
\end{align*}
and there exist
\(\fa'^{\circ}_{\Mark},\fa'^{\circ}_{0},\dots,\fa'^{\circ}_{k+m+m''}\)
such that
\begin{align*} &
\penv, \seq{x'}_{1,\dots,m}
 \pn 
[\seq{\fa} / \seq{y}]
\fa'
:\sty'\tr
\\&
\begin{alignedat}{20}
(&\fa'^{\circ}_{\Mark}&&,
\fa'^{\circ}_{0}&&,
\fa'^{\circ}_{1}&&,\dots&&,\fa'^{\circ}_{k}&&,
\fa'^{\circ}_{k+1}&&,
\dots&&,
\fa'^{\circ}_{k+m+m''}
&&)
\\
\beeq
(&\theta^{\Mark}\fa'_{\Mark}&&,
\theta^{0}\fa'_{\Mark}&&,
\theta^{1}\fa'_{\Mark}&&,
\ldots&&,
\theta^{k}\fa'_{\Mark}&&,
\theta^{\circ}\fa'_{1}&&,
\ldots&&,
\theta^{\circ}\fa'_{m+m''}
&&).
\end{alignedat}
\end{align*}

For any \(j\) and \(j'\in\{*,0,1,\dots,k,\circ\}\),
by the latter part of Lemma~\ref{lem:well-typedAndOccurrence},
\(\theta^{j}\fp = \theta^{j'}\fp\)
for
any formula \(\fp\) that occurs in
\begin{alignat*}{10}
\big(
\fm'_0(\fa'_0,\caa{\fa'}{m}{m''}),\,
&\fm'_1(\fa'_1,\caa{\fa'}{m}{m''})&&,
\ldots&&,\,
\fm'_m(\fa'_m,\caa{\fa'}{m}{m''})&&,
\\&
\fm'_{m+1}(\caa{\fa'}{m}{m''})&&,
\ldots&&,\,
\fm'_{m+m'}(\caa{\fa'}{m}{m''})&&\big).
\end{alignat*}
Especially, for any \(j\in\{*,0,1,\dots,k,\circ\}\), we have 
\begin{align}
\caa{\fa'^{\circ}}{k+m}{m''}
=\ &
(
\fa'^{\circ}_{0},
\fa'^{\circ}_{k+m+1},
\dots,
\fa'^{\circ}_{k+m+m''}
)
\notag\\\beeq\ &
(\theta^{0}\fa'_{\Mark},
\theta^{\circ}\fa'_{m+1},
\ldots,
\theta^{\circ}\fa'_{m+m''})
\notag\\=\ &
(\theta^{0}\fa'_{0},
\theta^{\circ}\fa'_{m+1},
\ldots,
\theta^{\circ}\fa'_{m+m''})
\tags{\text{by \cref{item:subOfSubstLemma}}}
\\=\ &
(\theta^{j}\fa'_{0},
\theta^{j}\fa'_{m+1},
\ldots,
\theta^{j}\fa'_{m+m''})
=
\theta^{j}\caa{\fa'}{m}{m''}.
\notag
\end{align}

Now, with \rnp{Tr-App}, we have
\begin{align*} &
(\fm^{\circ}_{\Mark},
\fm^{\circ}_{0},
\fm^{\circ}_{1},\dots,\fm^{\circ}_{k},
\fm^{\circ}_{k+1},
\dots,
\fm^{\circ}_{k+m},
\fm^{\circ}_{k+m+1},
\dots,
\fm^{\circ}_{k+m+m'}
)
\\=\ &
\begin{aligned}[t]
\big(&\fm'^{\circ}_\Mark(\fa'^{\circ}_\Mark,\caa{\fa'^{\circ}}{k+m}{m''}),\,
\fm'^{\circ}_0(\fa'^{\circ}_0,\caa{\fa'^{\circ}}{k+m}{m''}),\,
\\
&\begin{alignedat}[t]{10} &
\fm'^{\circ}_1(\fa'^{\circ}_1,\caa{\fa'^{\circ}}{k+m}{m''})
&&,\ldots,\,&&
\fm'^{\circ}_{k}(\fa'^{\circ}_{k},\caa{\fa'^{\circ}}{k+m}{m''})&&,
\\&
\fm'^{\circ}_{k+1}(\fa'^{\circ}_{k+1},\caa{\fa'^{\circ}}{k+m}{m''})
&&,\ldots,\,&&
\fm'^{\circ}_{k+m}(\fa'^{\circ}_{k+m},\caa{\fa'^{\circ}}{k+m}{m''})&&,
\\&
\fm'^{\circ}_{k+m+1}(\caa{\fa'^{\circ}}{k+m}{m''})
&&,\ldots,\,&&
\fm'^{\circ}_{k+m+m'}(\caa{\fa'^{\circ}}{k+m}{m''})&&\big)
\end{alignedat}
\end{aligned}
\\\beeq\ &
\begin{aligned}[t]
\big(&\theta^{\Mark}\fm'_{\Mark}(\theta^{\Mark}\fa'_\Mark,\theta^{\Mark}\caa{\fa'}{m}{m''}),\,
\theta^{0}\fm'_{\Mark}(\theta^{0}\fa'_{\Mark},\theta^{0}\caa{\fa'}{m}{m''}),\,
\\
&\begin{alignedat}[t]{10} &
\theta^{1}\fm'_{\Mark}
(\theta^{1}\fa'_{\Mark}
,\theta^{1}\caa{\fa'}{m}{m''})&&,
\ldots,\,&&
\theta^{k}\fm'_{\Mark}
(\theta^{k}\fa'_{\Mark}
,\theta^{k}\caa{\fa'}{m}{m''})&&,
\\&
\theta^{\circ}\fm'_{1}(\theta^{\circ}\fa'_{1},\theta^{\circ}\caa{\fa'}{m}{m''})&&,
\ldots,\,&&
\theta^{\circ}\fm'_{m}(\theta^{\circ}\fa'_{m},\theta^{\circ}\caa{\fa'}{m}{m''})&&,
\\&
\theta^{\circ}\fm'_{m+1}(\theta^{\circ}\caa{\fa'}{m}{m''})&&,
\ldots,\,&&
\theta^{\circ}\fm'_{m+m'}(\theta^{\circ}\caa{\fa'}{m}{m''})&&\big)
\end{alignedat}
\end{aligned}
\\=\ &
\begin{aligned}[t]
\big(&
\theta^{\Mark}\big(\fm'_{\Mark}(\fa'_\Mark,\caa{\fa'}{m}{m''})\big),\,
\theta^{0}\big(\fm'_{\Mark}(\fa'_{\Mark},\caa{\fa'}{m}{m''})\big),\,
\\
&\begin{alignedat}[t]{10} &
\theta^{1}\big(\fm'_{\Mark}
(\fa'_{\Mark}
,\caa{\fa'}{m}{m''})\big)&&,
\ldots,\,&&
\theta^{k}\big(\fm'_{\Mark}
(\fa'_{\Mark}
,\caa{\fa'}{m}{m''})\big)&&,
\\&
\theta^{\circ}\big(\fm'_{1}(\fa'_{1},\caa{\fa'}{m}{m''})\big)&&,
\ldots,\,&&
\theta^{\circ}\big(\fm'_{m}(\fa'_{m},\caa{\fa'}{m}{m''})\big)&&,
\\&
\theta^{\circ}\big(\fm'_{m+1}(\caa{\fa'}{m}{m''})\big)&&,
\ldots,\,&&
\theta^{\circ}\big(\fm'_{m+m'}(\caa{\fa'}{m}{m''})\big)&&\big)
\end{alignedat}
\end{aligned}
\\=\ &
(\theta^{\Mark}\fm_{\Mark},\theta^{0}\fm_{\Mark},
\theta^{1}\fm_{\Mark},
\ldots,
\theta^{k}\fm_{\Mark},
\theta^{\circ}\fm_{1},
\ldots,
\theta^{\circ}\fm_{m+m'}),
\end{align*}
as required.

\end{proof}

Now we are ready to prove Lemma~\ref{lem:subjectReduction}.
\begin{pfof}{Lemma~\ref{lem:subjectReduction}}
For the convenience of the proof,
we rename the metavariables
\(\fm, \fa, \fm_i, \fa_i\) with
\(\fm', \fa', \fm'_i, \fa'_i\):
so we suppose
\(\fm' \redd \fa'\)
and 
\(
\seq{x}_{1,\ldots,k}
\pn \fm':\Prop\tr (\fm'_\Mark,\fm'_0,\ldots,\fm'_k)
\), 
and prove that
we have
\(
\seq{x}_{1,\ldots,k}
\pn \fa':\Prop\tr (\fa'_\Mark,\fa'_0,\ldots,\fa'_k)
\)
and 
\(\fm'_i \eqdp \fa'_i\).
The proof proceeds by induction on \(\fm'\).

Let \(\fm'\) be of the form \(E[\fm'']\) where \(\fm''\) is a redex of \(\redd\).
The case where \(E\neq \hole\) can be easily proved by using induction hypothesis.\asd{todo}
So we consider only the case where \(E=\hole\).
Then we perform case analysis on \(\fm' \redd \fa'\),
but we focus only on the non-trivial case
where we use the substitution lemmas.

\myitem
Case where
\(\fm' \redd \fa'\)
is of the form
\[
F\,\fea_1\,\cdots\,\fea_{h'} \redd 
\esubst{\seq{x''}}{\seq{\fp}}
[\seq{e''} / \seq{z''}]
[\seq{\fm''} / \seq{y''}]\fm
\]
with the following conditions:
\begin{align*}
&(F\,\seq{w'}= \fm) \in D
\\&
\decomp(\NT(F))=(\seq{\sty''}_{1,\dots,h''},m, p)
\\&
\decompA(\seq{\alpha},\NT(F)) =
(\seq{\fm''},
\seq{\fp},
\seq{e''})
\\&
\decomparg(\seq{w'},\NT(F)) =
(\seq{y''}\COL\seq{\sty''},
 \seq{x''},
 \seq{z''})
\\&
\text{\(\seq{x''}\) do not occur in \(F\,\fea_1\,\cdots\,\fea_{h'}\).}
\end{align*}
By the last condition above,
we can assume \(\{x_i\}_{i}\cap
\{x''_i\}_{i} = \emptyset\).
%
%

Now there exist \(q\),
\(r_1 , \dots , r_{q+m+1}\),
\(\seq{\fa}_{1,\dots,q}\),
\(\seq{e}_{1,\dots,r_{q+m+1}}\)
that satisfy the following conditions,
where we write \(\seqe{i}\)
(or simply \(\seq{e}\) if \(i\) is clear)
for \(\seq{e}_{r_{i-1}+1,\dots,r_{i}}\)
(\(i=1,\dots,q+m+1\)) and define \(r_{0} \defe 0\):
\begin{align*}
&r_1 \le \dots \le r_{q+m+1}
\\&
(\fea_1,\dots,\fea_{h''})
=(
\seqe{1},\fa_1,
\dots,
\seqe{q},
\fa_q
)
\\&
(\fea_{h''+1},\dots,\fea_{h'})=
(
\seqe{q+1}, \fp_1,
\dots,
\seqe{q+m}, \fp_m,
\seqe{q+m+1}
)
\\&
p = r_{q+m+1} - r_{q}.
\end{align*}
Let \(\sty_i\) be the type of \(\fa_i\)
(i.e., \(\sty_i \defe \sty''_{r_{i}+i}\)).
Then we also have
\begin{align*}
\NT(F) =\ 
&\INT^{r_1} \to \sty_1 \to 
\cdots
\INT^{r_{q}-r_{q-1}} \to \sty_q \to 
\\&
\INT^{r_{q+1}-r_{q}} \to (\INT^{\maar{}} \to \Prop) \to 
\\&
\cdots
\INT^{r_{q+m}-r_{q+m-1}} \to
(\INT^{\maar{}} \to \Prop) \to 
\\&
\INT^{r_{q+m+1}-r_{q+m}} \to \Prop.
\end{align*}

In the derivation tree of
\[
\seq{x}_{1,\ldots,k}
\pn (\fm'=)\
F\,\fea_1\,\cdots\,\fea_{h'}
:\Prop\tr (\fm'_\Mark,\fm'_0,\ldots,\fm'_k),
\]
the leftmost path from the head position \(F\)
consists of:
(i) 
\rnp{Tr-VarF} at the leaf,
then 
(ii) 
repeated applications of 
either \rnp{Tr-App} or \rnp{Tr-AppI},
and then
(iii) 
repeated applications of 
either \rnp{Tr-AppG} or \rnp{Tr-AppI}.
More specifically,
at the leaf of \rnp{Tr-VarF} we have
\[
\envx \pn F:\NT(F)\tr
(F_0, F_0, (F_0)^{k}, F_1,\ldots,F_m)
\]
where \((F_0)^{k}\) denotes the sequence of length \(k\)
whose all components are \(F_0\).
Then by \rnp{Tr-AppI} we have
\begin{align}
&\envx \pn F\,\seqe{1}: \btype{\NT(F)}{r_1} \tr
\notag\\&\quad
(F_0\,\seqe{1}, F_0\,\seqe{1}, 
(F_0\,\seqe{1})^k,
F_1\,\seqe{1},\ldots,F_m\,\seqe{1}).
\notag
\end{align}
Then by \rnp{Tr-App} (and \rnp{Tr-AppI}) we have
\begin{align}
&
\envx \pn 
\fa_i : \sty_i \tr
(\fa_{i,\Mark},\fa_{i,0},\dots,\fa_{i, k + m'_i})
\quad(i=1,\dots,q)
\notag\\&
\envx \pn 
F\,\seqe{1}\,\fa_1\cdots\seqe{q}\,\fa_q
: \btype{\NT(F)}{r_{q}+q} \tr
\notag\\&\quad
\big(
F_0\,\seqe{1}\,(\fa_{1,\Mark},\caa{\fa_1}{k}{m'_1})
\cdots\seqe{q}\,(\fa_{q,\Mark},\caa{\fa_q}{k}{m'_q}),
\notag\\&\quad\phantom{\big(}
F_0\,\seqe{1}\,(\fa_{1,0},\caa{\fa_1}{k}{m'_1})
\cdots\seqe{q}\,(\fa_{q,0},\caa{\fa_q}{k}{m'_q}),
\notag\\&\quad\phantom{\big(}
\varf{F_0}{1}\,\seqe{1}\,(\fa_{1,1},\caa{\fa_1}{k}{m'_1})
\cdots\seqe{q}\,(\fa_{q,1},\caa{\fa_q}{k}{m'_q}),
\notag\\&\quad\phantom{\big(}
\ldots,
\varf{F_0}{k}\,\seqe{1}\,(\fa_{1,k},\caa{\fa_1}{k}{m'_1})
\cdots\seqe{q}\,(\fa_{q,k},\caa{\fa_q}{k}{m'_q}),
\notag\\&\quad\phantom{\big(}
F_1\,\seqe{1}\,(\caa{\fa_1}{k}{m'_1})
\cdots\seqe{q}\,(\caa{\fa_q}{k}{m'_q}),
\notag\\&\quad\phantom{\big(}
\ldots,
F_m\,\seqe{1}\,(\caa{\fa_1}{k}{m'_1})
\cdots\seqe{q}\,(\caa{\fa_q}{k}{m'_q})
\big)
\notag\\&
m'_i \defe \gar(\sty_i)
\quad(i=1,\dots,q)
\notag\\&
\decomp(\sty_i) = (\seq{\sty_i},m'_i,p_i)
\quad(i=1,\dots,q).
\notag
\end{align}
And then by \rnp{Tr-AppG} (and \rnp{Tr-AppI}) we have
\begin{align}
&
p^\circ_i \defe r_{q+m+1} - r_{q+i}
\quad(i=1,\dots,m)
\notag\\&
\fm^\circ_1 \defe
F\,\seqe{1}\,\fa_1\cdots\seqe{q}\,\fa_q\seqe{q+1}
\notag\\&
\fm^\circ_{i+1} \defe \fm^\circ_{i}\,\fp_{i}\,\seqe{q+i+1}
\quad(i=1,\dots,m)
\notag\\&
\fty_i \defe \btype{\NT(F)}{r_{q+i}+q+i}
\quad(i=1,\dots,m)
\notag\\&
\envx \pn \fm^\circ_{i}:(\INT^{\maar{}}{\to}\Prop) \,{\to}\, \fty_i\tr 
(\fm^\circ_{i,\Mark},\fm^\circ_{i,0},\ldots,\fm^\circ_{i,k+m+1-i})
\notag\\&
\tag{\(i=1,\dots,m\)}
\\&
\envx \pn \fm^\circ_{m+1}:\Prop\tr
(\fm^\circ_{m+1,\Mark},\fm^\circ_{m+1,0},\ldots,\fm^\circ_{m+1,k})
\notag\\&
\envx \pn \fp_{i}:{\INT}^{\maar{}}{\to}\Prop\tr (\fp_{i,\Mark},\fp_{i,0},\ldots,\fp_{i,k})
\tag{\(i=1,\dots,m\)}
\\&
\begin{aligned}[t]
\overline{\fp_{i,j}} \defe\ &
    \lambda \seq{z}_{1,\ldots,p^\circ_i}.\lambda \seq{w}_{1,\ldots,\maar{}}.
\\
   &\fm^\circ_{i,j}\,\seq{z}\,\seq{w} \, \lor
 \exists \seq{u}_{1,\ldots,\maar{}}.
    (\fm^\circ_{i,k+1}\,\seq{z}\,\seq{u}
    \land \fp_{i,j}\,\seq{u}\,\seq{w})
\end{aligned}
\notag\\&
\tag{\(i=1,\dots,m,\ j=\Mark,0,1,\dots,k\)}
\\&
\envx \pn 
\fm^\circ_{i+1}
(= \fm^\circ_{i}\,\fp_{i}\,\seqe{q+i+1}) :
\btype{\fty_i}{r_{q+i+1}-r_{q+i}} \tr
\notag\\&
\quad(
\overline{\fp_{i,\Mark}}\,\seq{e},
\overline{\fp_{i,0}}\,\seq{e},\dots,\overline{\fp_{i,k}}\,\seq{e},
\fm^\circ_{i,k+2}\,\seq{e},\dots,\fm^\circ_{i,k+m+1-i}\,\seq{e}
)
\notag\\&
\mspace{320mu}(i=1,\dots,m).
\notag
\end{align}
Then, for each \(i=2,\dots,m+1\), 
we have
\begin{alignat*}{20}
&(\fm^\circ_{i,\Mark}&&\,,\ &&
\fm^\circ_{i,0}&&\,,\ &&\ldots,\fm^\circ_{i,k}&&\,,\ &&
\fm^\circ_{i,k+1}&&\,,\ &&\ldots,\fm^\circ_{i,k+m+1-i}
&&\,)
\\=\ &
(
\overline{\fp_{i-1,\Mark}}\,\seq{e}&&\,,\ &&
\overline{\fp_{i-1,0}}\,\seq{e}&&\,,\ &&\dots,
\overline{\fp_{i-1,k}}\,\seq{e}&&\,,\ &&
\fm^\circ_{i-1,k+2}\,\seq{e}&&\,,\ &&\dots,\fm^\circ_{i-1,k+m+2-i}\,\seq{e}
&&\,)
\end{alignat*}
where \(\seq{e} = 
\seqe{q+i}\).
Hence, for each \(i=1,\dots,m\),
\begin{align*}
\fm^\circ_{i,k+1}
&=
\fm^\circ_{i-1,k+2}\,\seqe{q+i}
=
\fm^\circ_{i-2,k+3}\,\seqe{q+i-1}\,\seqe{q+i}
=
\dots
\\
&=
\fm^\circ_{1,k+i}\,\seqe{q+2}\cdots\seqe{q+i}
\\
&=
F_i\,\seqe{1}\,(\caa{\fa_1}{k}{m'_1})
\cdots\seqe{q}\,(\caa{\fa_q}{k}{m'_q})
\,\seqe{q+1}\cdots\seqe{q+i}
\end{align*}
where the last equality follows from
the calculation result of \rnp{Tr-App} above.
Also, for each \(i=2,\dots,m\) and \(j=\Mark,0,\dots,k\), we have
\begin{align*}
&\overline{\fp_{i,j}}
\,\seq{z}_{1,\ldots,p^\circ_i}\,\seq{w}_{1,\ldots,\maar{}}
\\\eqdp\ &
\overline{\fp_{i-1,j}}\,\seqe{q+i}
\,\seq{z}\,\seq{w} \, \lor
 \exists \seq{u}_{1,\ldots,\maar{}}.
    (\fm^\circ_{i,k+1}\,\seq{z}\,\seq{u}
    \land \fp_{i,j}\,\seq{u}\,\seq{w})
\end{align*}

Now, since \(\fm' = \fm^\circ_{m+1}\),
for each \(j=\Mark,0,\dots,k\), we have
\begin{align}
&\fm'_j\,\seq{w}_{1,\ldots,\maar{}}
= \fm^\circ_{m+1,j}\,\seq{w}
=
\overline{\fp_{m,j}}\,\seqe{q+m+1}\,\seq{w}
\notag\\\eqdp\ &
\overline{\fp_{m-1,j}}\,\seqe{q+m}
\,\seqe{q+m+1}\,\seq{w}
\notag\\&
\lor \exists \seq{u}_{1,\ldots,\maar{}}.
    (\fm^\circ_{m,k+1}\,\seqe{q+m+1}\,\seq{u}
    \land \fp_{m,j}\,\seq{u}\,\seq{w})
\notag\\\eqdp\ &
\overline{\fp_{m-2,j}}\,\seqe{q+m-1}
\,\seqe{q+m}
\,\seqe{q+m+1}\,\seq{w} 
\notag\\&
\lor \exists \seq{u}_{1,\ldots,\maar{}}.
    (\fm^\circ_{m-1,k+1}\,\seqe{q+m}
\,\seqe{q+m+1}\,\seq{u}
    \land \fp_{m-1,j}\,\seq{u}\,\seq{w})
\notag\\&
\lor \exists \seq{u}_{1,\ldots,\maar{}}.
    (\fm^\circ_{m,k+1}\,\seqe{q+m+1}\,\seq{u}
    \land \fp_{m,j}\,\seq{u}\,\seq{w})
\notag\\\eqdp\ & \cdots
\notag\\\eqdp\ &
\overline{\fp_{1,j}}\,\seqe{q+2}
\cdots
\seqe{q+m+1}\,\seq{w} 
\notag\\&
\lor \exists \seq{u}_{1,\ldots,\maar{}}.
    (\fm^\circ_{2,k+1}\,\seqe{q+3}\cdots\seqe{q+m+1}\,\seq{u}
    \land \fp_{2,j}\,\seq{u}\,\seq{w})
\notag\\&
\lor \cdots
\notag\\&
\lor \exists \seq{u}_{1,\ldots,\maar{}}.
    (\fm^\circ_{m,k+1}\,\seqe{q+m+1}\,\seq{u}
    \land \fp_{m,j}\,\seq{u}\,\seq{w})
\notag\\\eqdp\ &
 \fm^\circ_{1,j}\,\seqe{q+2}\cdots\seqe{q+m+1}\,\seq{w}
\notag\\&
\lor \exists \seq{u}_{1,\ldots,\maar{}}.
 (\fm^\circ_{1,k+1}\,\seqe{q+2}\cdots\seqe{q+m+1}\,\seq{u}
 \land \fp_{1,j}\,\seq{u}\,\seq{w})
\notag\\&
\lor \exists \seq{u}_{1,\ldots,\maar{}}.
    (\fm^\circ_{2,k+1}\,\seqe{q+3}\cdots\seqe{q+m+1}\,\seq{u}
    \land \fp_{2,j}\,\seq{u}\,\seq{w})
\notag\\&
\lor \cdots
\notag\\&
\lor \exists \seq{u}_{1,\ldots,\maar{}}.
    (\fm^\circ_{m,k+1}\,\seqe{q+m+1}\,\seq{u}
    \land \fp_{m,j}\,\seq{u}\,\seq{w})
\notag\\\eqdp\ &
 \fm^\circ_{1,j}\,\seqe{q+2}\cdots\seqe{q+m+1}\,\seq{w}
\notag\\&
\lor \textstyle\bigvee_{i=1}^{m}
 \exists \seq{u}_{1,\ldots,\maar{}}.
    (\fm^\circ_{i,k+1}\seqe{q+i+1}\cdots\seqe{q+m+1}\,\seq{u}
    \land \fp_{i,j}\,\seq{u}\,\seq{w})
\notag\\\eqdp\ &
 \fm^\circ_{1,j}\,\seqe{q+2}\cdots\seqe{q+m+1}\,\seq{w}
\notag\\&
\lor \textstyle\bigvee_{i=1}^{m}
 \exists \seq{u}_{1,\ldots,\maar{}}.
 (
\begin{aligned}[t]
&\begin{aligned}[t]
F_i\,&\seqe{1}\,(\caa{\fa_1}{k}{m'_1})
\cdots\seqe{q}\,(\caa{\fa_q}{k}{m'_q})
\\
&\seqe{q+1}\cdots\seqe{q+m+1}\,\seq{u}
\end{aligned}
\\
& \land \fp_{i,j}\,\seq{u}\,\seq{w}).
\end{aligned}
\label{eq:subRedMain1}
\end{align}

To calculate \(\fm^\circ_{1,j}\) and \(F_i\) above,
let us consider the rules of \(F_0,\dots,F_m\),
which are given by \rnp{Tr-Def} as follows.
Recall
\[
\decomparg(\seq{w'},\NT(F)) =
(\seq{y''}\COL\seq{\sty''},
\seq{x''},
 \seq{z''})
\]
and let
\begin{align}
&
    \seq{y''}\COL\seq{\sty''},
    \seq{z''}\COL\seq{\INT};\,
\seq{x''}_{1,\dots,m}
 \pn 
\fm:\Prop\tr (\fm_\Mark,\fm_0,\ldots,\fm_m)
\notag\\&
 \seq{y''}_i \defe (y''_{i,\Mark},y''_{i,0},\ldots,y''_{i,\gar(\sty''_i)})
\quad
 \seq{y''}^\circ_i \defe (y''_{i,0},\ldots,y''_{i,\gar(\sty''_i)})
\notag\\
\tag{\(i\in\set{1,\dots,h''}\) and \(\sty''_i\neq\INT\)}
\\&
 \seq{y''}_i \defe y''_i
\quad
 \seq{y''}^\circ_i \defe y''_i
\mspace{47mu}
(i\in\set{1,\dots,h''}\text{ and }\sty''_i=\INT).
\notag
\end{align}
Then we obtain
\begin{align*}
  \pn (F\,\seq{w'}= \fm)\tr\,
    &\set{F_0\,\seq{y''}_1\,\cdots\,\seq{y''}_{h''}\,\seq{z''}=\fm_\Mark}
\\
\cup&\set{F_i\,\seq{y''}^\circ_1\,\cdots\,\seq{y''}^\circ_{h''}\,\seq{z''}=\fm_i\mid
     i\in\set{1,\ldots,m}}.
\end{align*}
Recall that \(\sty_i \defe \sty''_{r_{i}+i}\) (\(i=1,\dots,q\)),
and let
\begin{align*}
&y_{i} \defe y''_{r_{i}+i}
\quad (i=1,\dots,q)
\\&
y_{i,j} \defe y''_{r_{i}+i,j}
\quad (i=1,\dots,q,\ j=*,0,\dots,\gar(\sty_i))
\\&
\seq{y}_i \defe \seq{y''}_{r_{i}+i}
= (y_{i,*},y_{i,0},\dots,y_{i,\gar(\sty_i)})
\\&
\seq{y}^\circ_i \defe \seq{y''}^\circ_{r_{i}+i}
= (y_{i,0},\dots,y_{i,\gar(\sty_i)}).
\end{align*}
Then let \(\seq{z}_{1,\dots,r_{q+m+1}}\)
be a sequence of variables of type \(\INT\)
that satisfies the following equations,
where we write \(\seqz{i}\)
(or simply \(\seq{z}\) if \(i\) is clear)
for \(\seq{z}_{r_{i-1}+1,\dots,r_{i}}\)
(\(i=1,\dots,q+m+1\)):
\begin{align*} &
(w'_1,\dots,w'_{h''}) =
(y''_1,\dots,y''_{h''}) =
(\seqz{1},y_1,\dots,\seqz{q},y_q)
\\&
(w'_{h''+1},\dots,w'_{h'}) =
(\seqz{q+1},x''_1,\dots,\seqz{q+m},x''_{m},\seqz{q+m+1})
\\&
(\seq{y''}_1,\dots,\seq{y''}_{h''},\seq{z''})
=(
\seqz{1},
\seq{y}_{1},
\dots,
\seqz{q},
\seq{y}_{q},
\seq{z}_{r_{q}+1,\dots,r_{q+m+1}}
)
\\&
(\seq{y''}^\circ_1,\dots,\seq{y''}^\circ_{h''},\seq{z''})
=(
\seqz{1},
\seq{y}^\circ_{1},
\dots,
\seqz{q},
\seq{y}^\circ_{q},
\seq{z}_{r_{q}+1,\dots,r_{q+m+1}}
).
\end{align*}

Now, for \(j=\Mark,0,\dots,k\), we have
\begin{align}
&\fm^\circ_{1,j}\,\seqe{q+2}\cdots\seqe{q+m+1}\,\seq{w}_{1,\ldots,\maar{}}
\notag\\\eqdp\ &
F_0\,\seqe{1}(\fa_{1,j},\caa{\fa_1}{k}{m'_1})
\cdots\seqe{q}(\fa_{q,j},\caa{\fa_q}{k}{m'_q})
\seqe{q+1}\cdots\seqe{q+m+1}\seq{w}
\notag\\\eqdp\ &
\big(
\big[ (\fa_{i',j},\caa{\fa_{i'}}{k}{m'_{i'}}) / \seq{y}_{i'} \big]_{i'=1}^{q}
\big[ e_{j'} / z_{j'} \big]_{j'=1}^{r_{q+m+1}}
\fm_\Mark
\big)
\,\seq{w}.
\label{eq:subRedMain2-1}
\end{align}
Also for each \(i=1,\dots,m\), we have
\begin{align}
&F_i\,\seqe{1}\,(\caa{\fa_1}{k}{m'_1})
\cdots\seqe{q}\,(\caa{\fa_q}{k}{m'_q})
\,\seqe{q+1}\cdots\seqe{q+m+1}\,\seq{u}
\notag\\\eqdp\ &
\big(
\big[ \caa{\fa_{i'}}{k}{m'_{i'}} / \seq{y}^\circ_{i'} \big]_{i'=1}^{q}
\big[ e_{j'} / z_{j'} \big]_{j'=1}^{r_{q+m+1}}
\fm_i
\big)\,\seq{u}.
\label{eq:subRedMain2-3}
\end{align}

Next, let us consider \(\fa'\).
Now we have
\begin{align*}
\fa'
&=
\esubst{\seq{x''}}{\seq{\fp}}
[\seq{e''} / \seq{z''}]
[\seq{\fm''} / \seq{y''}]\fm
\\&=
\esubst{\seq{x''}}{\seq{\fp}}
[\fa_{i'} / y_{i'}]_{i'=1}^{q}
[e_{j'} / z_{j'}]_{j'=1}^{r_{q+m+1}}
\fm.
\end{align*}
Recall
\begin{align*} &
 \seq{y''}\COL\seq{\sty''},\
 \seq{z''}\COL\seq{\INT};\
 \seq{x''}_{1,\dots,m}
 \pn 
\fm:\Prop\tr (\fm_\Mark,\fm_0,\ldots,\fm_m),
\\&
(\seq{y''}, \seq{z''})
=
(\seqz{1},y_1,\dots,\seqz{q},y_q,\seqz{q+1},\dots,\seqz{q+m+1}),
\end{align*}
and let
\begin{align*} &
\seq{y}:\seq{\sty};\ 
\seq{x''}_{1,\dots,m}
 \pn 
[e_{j'} / z_{j'}]_{j'=1}^{r_{q+m+1}}
\fm
:\Prop\tr
\\&\quad
(\fa'''_\Mark,\fa'''_0,\ldots,\fa'''_{m}),
\\&
\envx,\
 \seq{x''}_{1,\dots,m}
 \pn 
[\fa_{i'} / y_{i'}]_{i'=1}^{q}
[e_{j'} / z_{j'}]_{j'=1}^{r_{q+m+1}}
\fm
:\Prop\tr
\\&\quad
(\fa''_\Mark,\fa''_0,\ldots,\fa''_{k+m}).
\end{align*}
Then, by applying Lemma~\ref{lem:substitutionHigher}
\cref{item:mainOfSubstLemma},
and then by applying Lemma~\ref{lem:substitutionInt},
we obtain:
\begin{align}
\fa''_j
&\beeq
\big[ (\fa_{i',j},\caa{\fa_{i'}}{k}{m'_{i'}}) / \seq{y}_{i'} \big]_{i'=1}^{q}
\fa'''_\Mark
\notag\\&=
\big[ (\fa_{i',j},\caa{\fa_{i'}}{k}{m'_{i'}}) / \seq{y}_{i'} \big]_{i'=1}^{q}
\big[ e_{j'} / z_{j'} \big]_{j'=1}^{r_{q+m+1}}
\fm_\Mark
\tags[,]{j=\Mark,0,\dots,k}
\\
\fa''_{k+i}
&\beeq
\big[ \caa{\fa_{i'}}{k}{m'_{i'}} / \seq{y}^\circ_{i'} \big]_{i'=1}^{q}
\fa'''_i
\notag\\&=
\big[ \caa{\fa_{i'}}{k}{m'_{i'}} / \seq{y}^\circ_{i'} \big]_{i'=1}^{q}
\big[ e_{j'} / z_{j'} \big]_{j'=1}^{r_{q+m+1}}
\fm_i
\tags[.]{i=1,\dots,m}
\\
\label{eq:subRedMain3}
\end{align}

Now, for \(j=\Mark,0,\dots,k\), let
\begin{align}
\fa'_j=
\lambda \seq{w}_{1,\ldots,\maar{}}. 
\fa''_j\,\seq{w}
\notag 
\lor \textstyle\bigvee_{i=1}^{m}
 \exists \seq{u}_{1,\ldots,\maar{}}.
\Big(
\fa''_{k+i}\,\seq{u}
\land \fp_{i,j}\,\seq{u}\,\seq{w}
\Big).\\
\label{eq:subRedMain5}
\end{align}
Then by \rnp{Tr-ESub}, we have
\begin{align*}
&\envx \pn 
(\fa'=)
\esubst{\seq{x''}}{\seq{\fp}}
[\fa_{i'} / y_{i'}]_{i'=1}^{q}
[e_{j'} / z_{j'}]_{j'=1}^{r_{q+m+1}}
\fm
:\Prop\tr
\\&\quad
(\fa'_\Mark,\fa'_0,\ldots,\fa'_{k}).
\end{align*}
Also, by \cref{%
eq:subRedMain1,%
eq:subRedMain2-1,%
eq:subRedMain2-3,%
eq:subRedMain3,%
eq:subRedMain5},
we have \(\fm'_j \eqdp \fa'_j\)
for \(j=\Mark,0,\dots,k\),
as required.
\ifdraft\newpage
\asd{The following cases are shown only in the draft mode.}
\myitem
Case where 
\(\fm' \redd \fa'\)
is of the form
\[
\fm_1\lor\fm_2 \redd \fm'_1\lor\fm_2
\]
with
\(
\fm_1 \redd \fm'_1
\):

\myitem
Case where 
\(\fm' \redd \fa'\)
is of the form
\[
\fm_1\lor\fm_2 \redd \fm_1\lor\fm'_2
\]
with
\(
\fm_2 \redd \fm'_2
\):
This case is similar to the previous case.

\myitem
Case where 
\(\fm' \redd \fa'\)
is of the form
\[
e_1\le e_2 \land \fm\redd \FALSE\land \fm
\]
with
\(
\sem{\pST e_1:\INT}  >  \sem{\pST e_2:\INT}
\) and \(
(e_1\le e_2) \neq \FALSE
\):

\myitem
Case where 
\(\fm' \redd \fa'\)
is of the form
\[
e_1\le e_2 \land \fm\redd \fm
\]
with
\(
\sem{\pST e_1:\INT} \le \sem{\pST e_2:\INT}
\):

\fi
\end{pfof}

\fi
\end{document}